\documentclass[10pt,reqno,journal]{IEEEtran}

\usepackage{amssymb}
\usepackage{xcolor}
\usepackage{cite}
\usepackage{subcaption}
\usepackage{epsfig}
\usepackage[amssymb]{SIunits}
\usepackage{graphicx}
\usepackage{epstopdf}
\usepackage{amsmath}
\usepackage{amsthm}
\usepackage[amssymb]{SIunits}
\allowdisplaybreaks

\newcommand{\ncom}{\newcommand}
\ncom{\beqn}{\begin{eqnarray*}} 
\ncom{\eeqn}{\end{eqnarray*}}
\ncom{\beq}{\begin{eqnarray}} 
\ncom{\eeq}{\end{eqnarray}}
\newtheorem{theorem}{Theorem}
\newtheorem{assumption}{Assumption}
\newtheorem{proposition}{Proposition}

\newtheorem{definition}{Definition}
\newtheorem{remark}{Remark}
\newtheorem{lemma}{Lemma}

\ncom{\R}{I\!\!R}
\newcommand{\norm}[1]{\vert\vert#1\vert\vert}

\begin{document}
\title{
Event-triggered Control for Nonlinear Systems with Center Manifolds}
\author{Akshit Saradagi$^{1}$, Vijay Muralidharan$^{2}$, Arun D. Mahindrakar$^{3}$,~\IEEEmembership{Senior Member,~IEEE} and Pavankumar Tallapragada$^{4},~\IEEEmembership{Member,~IEEE}$
  \thanks{$^{1}$Akshit and $^{3}$Arun are with the Department of Electrical Engineering, Indian Institute of Technology Madras, Chennai, India (emails: akshitsaradagi@gmail.com, arun\_dm@iitm.ac.in)} 
  \thanks{$^{2}$Vijay is with the Department of Electrical Engineering, Indian Institute of Technology Palakkad, Palakkad, India (email: vijay@iitpkd.ac.in)} 
  \thanks{$^{4}$Pavankumar Tallapragada is with the Department of Electrical Engineering, Indian Institute of Science, Bengaluru, India
        (email: pavant@iisc.ac.in)}
}
\maketitle
\begin{abstract}
In this work, we consider the problem of event-triggered implementation of control laws designed for the local stabilization of nonlinear systems with center manifolds. We propose event-triggering conditions which are derived from a local input-to-state stability characterization of such systems. The triggering conditions ensure local ultimate boundedness of the trajectories and the existence of a uniform positive lower bound for the inter-event times. The ultimate bound can be made arbitrarily small, but by allowing for smaller inter-event times. Under certain assumptions on the controller structure, local asymptotic stability of the origin is also guaranteed. Two sets of triggering conditions are proposed, that cater to the cases where the exact center manifold and only an approximation of the center manifold is computable. The closed-loop system exhibits some desirable properties when the exact knowledge of the center manifold is employed in checking the triggering conditions. Three illustrative examples that explore different scenarios are presented and the applicability of the proposed methods is demonstrated. The third example concerns the event-triggered implementation of a position stabilizing controller for the open-loop unstable Mobile Inverted Pendulum (MIP) robot.     
\end{abstract}

\section{Introduction}
\label{introduction}
Our view of what constitutes as resource in the implementation of control laws has evolved over the years. In many optimal control formulations, the control effort or control energy captures the sensing, computation and actuation costs and is traded with closed-loop performance. The focus on large-scale multi-agent systems in the recent years has revealed that, resources such as CPU time and wireless channel bandwidth cannot be overlooked, as they are shared and costly \cite{iot, networked}. In response, various resource-aware implementations have emerged, that look to utilize the resources judiciously, while guaranteeing pre-specified performance. Event-triggered control \cite{tab1, Tab_CDC, macej} is one such resource-aware technique which has gained popularity in the recent years and presents an alternative to time-triggered control. 

In event-triggered control, the control loop is closed when certain events occur in the system and not periodically as in time-triggered control. As the event instants are implicitly defined and the inter-event times are aperiodic, non-existence of Zeno behaviour (existence of a uniform positive lower bound for the inter-event times) must be ensured. Many event-triggered implementations require continuous monitoring of the states to detect the occurrence of an event. As this is not feasible in digital implementations, solutions are being explored where the occurrence of an event needs to be checked only periodically \cite{periodic_lin},\cite{borgers_periodic}. Often, controllers are designed first and then event-triggering conditions are designed. This sequential design procedure is called emulation-based approach. Recently, control and implementation co-design is being investigated, where both are designed parallelly \cite{abdelrahim_codesign}\cite{codesign}. 

Although research in event-triggered control has investigated a wide variety of settings, as surveyed in \cite{macej} and \cite{survey1}, the case of event-triggered control of nonlinear systems with center manifolds has not been looked at so far. Center manifold analysis is a crucial design and analysis tool for nonlinear systems with degenerate equilibria and is widely used in the areas of control theory, bifurcation theory and multi-scale modelling, among many others \cite{carr},\cite{gucken1}. We realised the need for new research in this direction when analysing the applicability of existing results for the Mobile Inverted Pendulum (MIP) robot\cite{vijay_TCST}. Control of a tethered satellite system is another example where a controller is designed in the presence of a center manifold \cite{tethered}. In this work, we present our solution to the problem of event-triggered control of nonlinear systems with center manifolds.

The work presented here builds on the preliminary findings from our recent work \cite{CDC_2020}, where an explicit Lyapunov-based characterisation of local input-to-state stability (LISS) for nonlinear systems with center manifolds was derived, assuming a general controller structure proposed in \cite{Hamzi_Krener}. This LISS characterization is used in the present work to design event-triggering conditions. In \cite{CDC_2020}, we also identified hurdles in the way of designing event-triggered implementation, namely, requirement of the exact knowledge of the center manifold in checking the triggering conditions and the non-applicability of the existing ISS-based results for nonlinear systems \cite{tab1}, \cite{Postoyan} for a large class of systems under consideration. 

For the general class of nonlinear systems, two approaches to the design of event-triggered control are popular in literature. The approach assuming input-to-state stability of the closed-loop system with respect to measurement errors \cite{tab1, Postoyan, Girard_dynamic, wang_lemmon, NOZARI}, guarantees asymptotic stability of the origin. However, the non-existence of Zeno behaviour can be ensured only under the assumption that the composition of the comparison functions involved in the ISS characterization is Lipschitz continuous over compact sets. This assumption is not satisfied for a large subset of the systems under consideration. Under the assumption of only asymptotic tracking, the work \cite{pavan} proposes event-triggered tracking (which can be adapted for stabilization) for nonlinear systems and guarantees Zeno-free uniform ultimate boundedness of the tracking error. For the systems under consideration, the Lyapunov analysis is essentially local and the results from \cite{pavan} are not directly applicable.

The main contributions of this work are the following. In this work, event-triggered implementation of control laws designed for local stabilization of nonlinear systems with center manifolds is being investigated for the first time. The proposed methods ensure Zeno-free local ultimate boundedness of the trajectories, with an ultimate bound that can be made arbitrarily small. Under some assumptions on the controller structure, Zeno-free asymptotic stability of the origin is also ensured. For a subclass of the systems under consideration, the LISS characterization meets the sufficient conditions from \cite{tab1} and \cite{Postoyan} and triggering conditions proposed in these works can be employed. For systems which do not fall in this class, which include many practical systems, such as the MIP robot, these triggering conditions cannot be used. The triggering conditions presented in this work cater to this class of systems and this forms our main contribution. As in \cite{pavan}, for a large class of systems under consideration, the triggering conditions presented in this work ensure ultimate boundedness of the trajectories, to an ultimate bound that can be made arbitrarily small. However, under some assumptions on the controller structure, Zeno-free asymptotic stability of the origin is also ensured. The proposed triggering conditions cater to both the cases where the exact and the approximate knowledge of the center manifold is available. Simulation results for three illustrative examples are presented, that explore the different scenarios being considered in this work. The third example concerns the event-triggered position stabilization of an open-loop unstable MIP robot.

The organization of this work is as follows. We begin with notations and preliminaries in Section \ref{preliminaries}. In Section \ref{formulation_4}, we recollect important results from center manifold theory, setup the problem and summarize the findings from our work in \cite{CDC_2020} on LISS of nonlinear systems with center manifolds. This LISS characterization is used in designing event-triggered control in Section \ref{event_4}, for systems with exactly computable center manifolds. In Section \ref{need}, possibility of designing triggering conditions that do not require the exact knowledge of the center manifold is explored. In section \ref{example_4}, we present triggering conditions utilizing the approximate knowledge of the center manifold and draw a comparison with the case where the exact knowledge of the center manifold is assumed. In subsection \ref{C5_position}, using the triggering conditions presented in section \ref{example_4}, event-triggered implementation of a position stabilizing controller for an MIP robot is presented, following which concluding remarks are made in section \ref{conclusions}. 
\section{NOTATIONS AND PRELIMINARIES}\label{preliminaries}
We denote by $\mathbb{R}$, the set of real numbers and by $\bar{\mathbb{R}}_{+}$ the set of non-negative real numbers. Given two vectors $y$ and $w$, we denote by $(y;w)$ the concatenation of the two vectors $[y^{\top}\;w^{\top}]^{\top}$. We denote by $|\cdot|$, the absolute value of a real number and by $\norm{\cdot}$, the Euclidean norm of a vector or the induced 2-norm of a matrix, depending on the argument. An $n \times n$ identity matrix is denoted by $\mathbb{I}_n$. Given a matrix $A \in \mathbb{R}^{n\times n}$, $A \succ 0$ denotes that $A$ is a symmetric positive definite matrix. A continuous function $\alpha: [0,a)\rightarrow [0,\infty)$ is said to be a class-${\cal K}$ function if $\alpha(0)=0$ and it is strictly increasing.  The notations $f(x)=\mathcal{O}(\norm{x}^p)$ and $f(x) \in \mathcal{O}(\norm{x}^p)$, $p \in \mathbb{R}$ denote that $|f(x)|\leq c_1\norm{x}^p$ for all $x$ such that $\norm{x}<\epsilon$, for some $c_1, \epsilon>0$.
\begin{definition}[Local input-to-state stability (LISS) \cite{huajuan}]\label{Local_Input}
The system $\dot{x}=f(x,d),\; x \in \mathbb{R}^n \text{\;and\;} d \in \mathbb{R}^m$ with $f$ being locally Lipschitz and $f(0,0)=0$, is said to be locally input-to-state stable in the domain $D_x\subset\mathbb{R}^n$ with respect to input $d$ in the domain $ D_d\subset\mathbb{R}^m$, if there exists a Lipschitz continuous function $V : D_x \rightarrow \mathbb{\bar{R}}_+$ and class-${\cal K}$ functions $\alpha_1, \alpha_2, \alpha_3$ and $\beta$ such that
\begin{align*}
\alpha_1(\norm{x})\leq V(x)\leq\alpha_2({\norm{x}})
\end{align*} 
and
\begin{align*}
\zeta^{\top}f(x,d) \leq -\alpha_3(\norm{x})+\beta(\norm{d})
\end{align*}
hold for all $x \in D_x$, $d \in D_d$ and $\zeta \in \partial V(x)$, the subdifferential \cite{convex_funda} of $V$ at $x$. The function $V$ satisfying the above conditions is called an LISS Lyapunov function.
\end{definition}

\section{Problem setup}\label{formulation_4}
In this section, we introduce the systems under consideration - nonlinear systems with center manifolds, review the control design techniques for such systems and setup the problem of designing event-triggered control with the most general controller structure.

Consider the nonlinear dynamical system
\begin{align}
\dot{x}=f(x,u), \;\; x(t_0)=x_0
\label{forced}
\end{align}
where $x\in \mathbb{R}^n$, $u \in \mathbb{R}^m$ and $f : \mathbb{R}^n \times \mathbb{R}^m \rightarrow \mathbb{R}^n$ is a twice continuously differentiable function with $f(0,0)=0$. The Taylor series expansion of the function $f$ about $x=0$ and $u=0$ yields
\begin{align}
\dot{x}=Ax+Bu+\tilde{f}(x,u)
\label{lin_forced}
\end{align}
where $A \in \mathbb{R}^{n \times n}$, $B \in \mathbb{R}^{n \times m}$ and $\tilde{f}(x,u)$ constitutes the higher-order terms and satisfies
\begin{equation}
\tilde{f}(0,0)=0, \;\; \frac{\partial{\tilde{f}}}{\partial{x}}(0,0)=0 \;\; \text{and} \;\; \frac{\partial{\tilde{f}}}{\partial{u}}(0,0)=0.
\label{conditions_1}
\end{equation}
In this work, we focus on nonlinear systems whose linearized models have uncontrollable modes on the imaginary axis. For such systems, there exists a linear transformation $x=T(y;z), T \in \mathbb{R}^{n \times n}$ such that system \eqref{lin_forced} is transformed into the following form
\begin{equation}
\begin{aligned}
\dot{y} &= A_1y + \tilde{g}_1(y,z,u)\\
\dot{z} &= A_2z + B_2u + \tilde{g}_2(y,z,u)
\label{nocenter}
\end{aligned}
\end{equation}
where $A_1 \in \mathbb{R}^{k\times k}$, $A_2  \in \mathbb{R}^{(n-k)\times(n-k)}$, $B_2 \in \mathbb{R}^{(n-k) \times m}$, $\tilde{g}_1, \tilde{g}_2$ are the nonlinearities, the real parts of the eigenvalues of $A_1$ are zero and the pair $(A_2,B_2)$ is controllable. The functions $\tilde{g}_1(y,z,u)$ and $\tilde{g}_2(y,z,u)$ satisfy 
\begin{equation}
\begin{aligned}
\tilde{g_i}(0,0,0)&=0, &\frac{\partial{\tilde{g_i}}}{\partial{y}}(0,0,0)=0, &\\ \frac{\partial{\tilde{g_i}}}{\partial{z}}(0,0,0)&=0, &\frac{\partial{\tilde{g_i}}}{\partial{u}}(0,0,0)=0, &\text{\;\;for\;\;} i=1,2.
\label{conditions_2} 
\end{aligned}
\end{equation}
This is easily inferred from the properties of function $\tilde{f}$ in \eqref{conditions_1}.

Suppose we design a controller $K(y,z)$ for system \eqref{nocenter} and wish to assess the local stability of the origin $(y,z)=0$ of the closed-loop system. The control $u$ cannot influence the eigenvalues of the matrix $A_1$, even though the poles of $A_2$ can be arbitrarily placed. The indirect method of Lyapunov through linearization cannot be used and results from center manifold theory have to be invoked. This theory provides a model-reduction technique to determine the stability of the closed-loop system by assessing the stability of a reduced system, which governs the dynamics on an invariant manifold called the center manifold. If the dynamics on the center manifold is stable (unstable), then the overall system is stable (unstable). In \cite{Hamzi_Krener}, \cite{shastry_unli}, \cite{Liaw}, systematic control design procedures were presented for nonlinear systems with center manifolds. In the following subsections, we review important results from center manifold theory and the control design approach from \cite{Hamzi_Krener}.
\subsection{Choice of a controller}\label{choice_of_a_controller}
The origin of the $z$-subsystem in \eqref{nocenter} can be locally asymptotically stabilized by a controller $u=K_{11}z, K_{11} \in \mathbb{R}^{m \times (n-k) }$, under the assumptions of stabilizability of the pair $(A_2,B_2)$. The feedback matrix $K_{11}$ indirectly determines both the center manifold and the dynamics on the center manifold. Stabilizability, while limiting the freedom in specifying the performance of the $z$-subsystem, also limits the indirect influence of the matrix $K_{11}$ on the centre manifold structure and the stability of the center manifold.  This justifies the assumption of controllability of the pair $(A_2,B_2)$ rather than stabilizability of the pair $(A_2,B_2)$. 

When the overall system cannot be stabilized by just $u=K_{11}z$, feedback of the form $u=K_{11}z+K_{12}y$ must be employed with $K_{12}$ chosen to stabilize the dynamics on the center manifold. This may also be inadequate, following which an additive pseudo-control $K_n(y)$ must be introduced to yield $u=K_{11}z+K_{12}y+K_n(y)$, as in \cite{Hamzi_Krener}. The pseudo-control $K_n:\mathbb{R}^k\rightarrow\mathbb{R}^m$ is a continuously differentiable nonlinear function of $y$ and is chosen to stabilize the dynamics on the center manifold. 
 In the rest of the work, we use the control structure $u=K(y;z)=K_{11}z+K_{12}y+K_n(y)$, as this is the most general form of the controller that can be used in the stabilization of nonlinear systems with center manifolds. Denoting $A_2+B_2K_{11}$ by $A_K$, we arrive at
\begin{equation}
\begin{aligned}
\dot{y} &= A_1y+\tilde{g}_1(y,z,K(y;z)) \\
\dot{z} &= A_Kz+B_2K_{12}y + B_2K_n(y)+\tilde{g}_2(y,z,K(y;z)). 
\label{center}
\end{aligned}
\end{equation}
\subsection{Background from Center Manifold Theory} \label{center_intro}
At this juncture, we recall some important definitions and results from center manifold theory.
%
For the next few results from center manifold theory to hold, the cross-coupling linear term, $B_2K_{12}y$ between the $y$ and $z$ subsystems in \eqref{center} must be eliminated. The cross-coupling term is eliminated using the change of variables $v \triangleq z-Ey$, where the matrix $E$ is found by solving the equation
\begin{align}
A_KE-EA_1+B_2K_{12}=0.
\label{finde}
\end{align}
As the sum of any eigenvalue of $A_1$ and any eigenvalue of $A_K$ is non-zero, we use from \cite{Liaw}, the result which states that there exists a unique matrix $E \in \mathbb{R}^{(n-k)\times k}$ such that \eqref{finde} is satisfied. Using this change of variables and the notations
\begin{align*}
g_1(y,v+&Ey,K(y;v+Ey))=\tilde{g}_1(y,v+Ey,K(y;v+Ey)) \\
g_2(y,v+&Ey,K(y;v+Ey))=B_2K_n(y)+\tilde{g}_2(y,v+Ey,\\
&K(y;v+Ey))-E\tilde{g}_1(y,v+Ey,K(y;v+Ey))
\end{align*}
we arrive at
\begin{equation}
\begin{aligned}
\dot{y} &= A_1y + g_1(y,v+Ey,K(y;v+Ey))\\
\dot{v} &= A_Kv + g_2(y,v+Ey,K(y;v+Ey)).
\label{final_tranf}
\end{aligned}
\end{equation}
Again, it can be easily checked that $g_1$ and $g_2$ also satisfy conditions \eqref{conditions_2}. 
\begin{definition}[Center Manifold]
For the dynamical system \eqref{final_tranf}, a manifold $v=h(y)$ is called a center manifold, if 
\begin{equation}
h(0)=0 \text{\;\;and\;\;} \dfrac{\partial h}{\partial y}(0)=0
\label{center_quad}
\end{equation}
and $v(0)=h(y(0))$ implies $v(t)=h(y(t)),~ \forall\;t\in[0,\infty)$.
\end{definition}

The center manifold is an invariant manifold of system \eqref{final_tranf}. The center manifold $h(y)$ satisfying \eqref{center_quad} is found by solving the partial differential equation
\begin{equation}
\begin{aligned}
0=& A_Kh(y)+g_2(y,h(y)+Ey,K(y;h(y)+Ey))\\
&-\frac{\partial h(y)}{\partial y}(A_1y + g_1(y,h(y)+Ey,K(y;h(y)+Ey))).
\label{partial} 
\end{aligned} 
\end{equation}
\begin{theorem}[Existence of a center manifold \cite{khalil}]\label{existence}
Consider the system \eqref{final_tranf}. If $g_1(y,v)$ and $g_2(y,v)$ are twice continuously differentiable and satisfy conditions in \eqref{conditions_2}, with all eigenvalues of $A_1$ having zero real parts and all eigenvalues of $A_K$ having negative real parts, then there exists a constant $\delta>0$ and a continuously differentiable function $h(y)$, defined for all $\norm{y}\leq \delta$, such that $v=h(y)$ is a center manifold for the system \eqref{final_tranf}.
\label{center_manifold_theorem}
\end{theorem}

System \eqref{final_tranf} satisfies the hypothesis of Theorem \ref{existence} for the existence of a $k$-dimensional center manifold $v=h(y)$ and the dynamics on the center manifold is governed by
\begin{align}
\dot{y}=A_1y+g_1(y,h(y)+Ey,K(y;h(y)+Ey))
\label{dynamics_center}
\end{align}
which is referred to as the reduced system. We next recall a result, known popularly as the \textit{Reduction Theorem}.
\begin{theorem}[Reduction Theorem \cite{khalil}]\label{red_theorem}
Under the assumptions of Theorem \ref{center_manifold_theorem}, if the origin $y=0$ of the reduced system \eqref{dynamics_center} is locally asymptotically stable (unstable), then the origin of the full system \eqref{final_tranf} is locally asymptotically stable (unstable).
\end{theorem}
%
\subsection{Event-based control and measurement errors} \label{measurement}
In event-triggered implementation, the control is updated at discrete instants $t_i$, $i \geq 0$, called the event times. Between two events, in the interval $[t_i, t_{i+1})$, the control is held constant to $u=K(y(t_i);v(t_i)+Ey(t_i))$ and measurement and actuation errors are introduced to account for the sample and hold nature of the implementation. With the measurement error $e_y(t) \triangleq y(t_i)-y(t)$ and $e_v(t) \triangleq v(t_i)-v(t)$, the control can be rewritten as $u=K(y+e_y;v+e_v+E(y+e_y))$. Under this implementation, the closed-loop system is
\begin{equation}
\begin{aligned}
\dot{y} =& A_1y+g_1(y,v+Ey,u) \\
\dot{v} =& A_Kv +B_2K_{11}e_v+B_2(K_{12}+EK_{11})e_y\\
&+g_2(y,v+Ey,u).
\label{with_error}
\end{aligned}
\end{equation}
For further analysis of the system, we introduce a new coordinate $w=v-h(y)$ (the need for which is explained in section \ref{need}). With control $u=K(y+e_y,w+h(y)+e_v+E(y+e_y))$ and $K_1=[K_{12} + E K_{11} \;\; K_{11}]$ and $e=(e_y,e_v)$, the dynamics \eqref{with_error} in the $(y;w)$ coordinates is found to be
\begin{align}
\dot{y} &= A_1y+g_1(y,w+h(y)+Ey,u)\nonumber\\
\dot{w} &= A_K(w+h(y))+B_2K_1e+ g_2(y,w+h(y)+Ey,u)\nonumber\\
&\hspace{1cm}-\dfrac{\partial h(y)}{\partial y}(A_1y +g_1(y,w+h(y)+Ey,u). \label{two}
\end{align}
Subtracting \eqref{partial} from \eqref{two} and using the following notations
\begin{align*}
N_1(y,w,e) \triangleq &g_1(y,w+h(y)+Ey,K(y+e_y;w+h(y)\\
&+e_v+E(y+e_y)))\\
&-g_1(y,h(y)+Ey,K(y;h(y)+Ey))\\
N_2(y,w,e) \triangleq &g_2(y,w+h(y)+Ey,K(y+e_y;w+h(y)\\
&+e_v+E(y+e_y)))-g_2(y,h(y)+Ey,\\
&K(y;h(y)+Ey))-\dfrac{\partial h(y)}{\partial y}(A_1y+g_1(y,w\\
&+h(y)+Ey,K(w+h(y)+e))-(A_1y\\
&+g_1(y,h(y)+Ey,K(y;h(y)+Ey))))
\end{align*}  
we obtain
\begin{equation}
\begin{aligned}
\dot{y} =& A_1y+g_1(y,h(y)+Ey,K(y;h(y)+Ey)) \\
& \hspace{4.75cm} + N_1(y,w,e) \\
\dot{w} =& A_Kw+B_2K_1e+N_2(y,w,e).  
\label{final}
\end{aligned}
\end{equation}
By substitution, we see that for $i=1, 2$
\begin{align}
N_i(y,0,0)=0, \frac{\partial N_i}{\partial w}(0,0,0)=0 \text{\;and\;} \frac{\partial N_i}{\partial e}(0,0,0)=0.
\label{crucial}
\end{align}
When $N_i$ satisfy conditions \eqref{crucial}, there exists a constant $\delta_{yw}>0$ such that, in the set 
\begin{align}
S=\{(y;w) \; | \; \norm{(y;w)}< \delta_{yw}\}
\label{delta}
\end{align}
we have for $i=1,2$,
\begin{align}
\norm{N_i}\leq k_i\norm{(w;e)}\leq k_i(\norm{w}+\norm{e})
\label{notice_difference}
\end{align}
where the constants $k_i>0$ can be made arbitrarily small by decreasing $\delta_{yw}$. Note that the stability properties of system \eqref{final} are the same as that of system \eqref{forced} in view of the sequence of smooth transformations $w=v-h(y)$, $v=z-Ey$ and $x=T(y;z)$ relating the two systems.

In this work, our design of event-triggered implementation uses the ISS based approach proposed in \cite{tab1} and generalized in \cite{Postoyan}. Verifying that the closed-loop system is ISS with respect to measurement errors forms the first step of this approach. This led us to the question : Does a controller that locally asymptotically stabilizes \eqref{final_tranf} also render \eqref{final} locally input-to-state stable with respect to measurement errors? 
Although the result \cite[Lemma 1]{sontag} answers this question in the affirmative, this qualitative result is useful for event-triggered control only when the class-${\cal K}$ function $\alpha_3$ and $\beta$ that characterize local input-to-state stability in Definition \ref{Local_Input} are known. Motivated by the unavailability of such explicit characterization of input-to-state stability in terms of comparison functions for nonlinear systems with center-manifolds and with the intent of designing event-triggered control, we investigated the local ISS property of the systems under consideration in our preliminary work \cite{CDC_2020}. In the next subsection, we recall an important result from this work.
\subsection{Local input-to-state stability of nonlinear systems with center manifolds}\label{input_to_state_stability}

The Proposition presented in this subsection establishes that a controller that locally asymptotically stabilizes system \eqref{final_tranf}, renders \eqref{final} LISS with respect to measurement errors. In Definition \ref{Local_Input} of local input-to-state stability, it is only required that $V$ be a Lipschitz continuous function. It is well know that a system is input-to-state stable if and only if there exists an ISS Lyapunov function\cite{sontag_coprime}, which is a continuously differentiable function. However, a continuously differentiable function is not necessary to establish ISS \cite{ISS_main}. In the following proposition from \cite{CDC_2020}, LISS of system \eqref{final} is shown by employing a nonsmooth Lyapunov function.
\begin{proposition}[\cite{CDC_2020}]\label{main_result}
Under the assumptions of Theorem \ref{center_manifold_theorem}, if the origin $y=0$ of the reduced system \eqref{dynamics_center} is locally asymptotically stable, then the overall system \eqref{final} is locally input-to-state stable with respect to the error $e=(e_y;e_v)$. 
\end{proposition}
The proof involved the construction of a nonsmooth ISS Lyapunov function $V:\mathbb{R}^n\rightarrow\mathbb{\bar{R}}_+$ 
\begin{align}
V(y,w)=V_1(y)+\sqrt{w^{\top}Pw}.
\label{main_lyapunov_fun}
\end{align}
By the converse Lyapunov theorem, the local asymptotic stability of the equilibrium point of the center manifold dynamics implies the existence of a continuously differentiable Lyapunov function $V_1 :\mathbb{R}^k\rightarrow \mathbb{\bar{R}}_+$ and class-${\cal K}$ functions $\alpha_4, \alpha_5$ such that 
\begin{align*}
\dot{V}_1&=\dfrac{\partial V_1}{\partial y}(A_1y+g_1(y,h(y)+Ey,K(y;h(y)+Ey)))\\
&\leq -\alpha_4(\norm{y}), \\
&\hspace{2cm} \left\lvert\left\lvert\dfrac{\partial V_1}{\partial y}\right\lvert\right\lvert\leq\alpha_5(\norm{y})\leq k_v
\end{align*}
for some $k_v>0$, in a neighbourhood of the origin. Also, since the matrix $A_K$ is Hurwitz, for every $Q\succ0$ there exists a unique $P\succ0$ such that $A_K^{\top}P+PA_{K}=-Q$. 
Taking the time derivative of $V$ along the trajectories of the system \eqref{final}, it was shown that
\begin{equation}
\begin{aligned}
\zeta^{\top}(\dot{y};\dot{w})&\leq\underbrace{-\alpha_4(\norm{y})-(1-s_f)\frac{\lambda_{min}(Q)}{2\sqrt{\lambda_{max}(P)}}\norm{w}}_{-\alpha_D(\norm{(y;w)})}\\
&\hspace{0.25cm}+\underbrace{\left(k_vk_1+k_2\frac{\lambda_{max}(P)}{\sqrt{\lambda_{min}(P)}}+\frac{\norm{PB_2K_1}}{\sqrt{\lambda_{min}(P)}}\right)\norm{e}}_{\beta_G(\norm{e}) }\\
&\hspace{0.5cm}+\gamma(\norm{w})
\end{aligned}
\label{LISS_Char}
\end{equation}
where $\zeta \in \partial V(x)$ and $s_f \in (0,1)$. The constant $\delta_{yw}$ defining the set $S$ in \eqref{delta} is chosen such that the constants $k_1$ and $k_2$ in \eqref{notice_difference} lead to $\gamma(\norm{w})=\left(k_vk_1+k_2\frac{\lambda_{max}(P)}{\sqrt{\lambda_{min}(P)}} -\frac{s_f \lambda_{min}(Q)}{2\sqrt{\lambda_{max}(P)}}\right)\norm{w}\leq 0$. The functions $\alpha_D(\norm{(y;w)})$ and $\beta_G(\norm{e})$ are class-${\cal K}$ functions of $\norm{(y;w)}$ and $\norm{e}$ respectively. From Definition \ref{Local_Input}, \eqref{LISS_Char} implies that
the origin of system \eqref{final} is locally input-to-state stable with respect to the error $e$.

Proposition \ref{main_result} generalizes the \textit{Reduction Theorem} (Theorem \ref{red_theorem}), as in the absence of the error $e$, local asymptotic stability of the overall system is recovered.
As mentioned before, the stability properties of system \eqref{final} are the same as that of system \eqref{forced}, in view of the sequence of smooth transformations relating the two systems. 
\begin{remark}
The work in \cite{carr} consolidated earlier work on center manifold theory and presented proofs of Theorems \ref{existence} and \ref{red_theorem} using the contraction mapping principle. In \cite{khalil}, a Lyapunov based proof of the \textit{Reduction Theorem} was presented for the first time using a nonsmooth Lyapunov function. The structure of the ISS Lyapunov functions chosen in our work is inspired by the nonsmooth Lyapunov function in \cite{khalil}. In \cite{Hamzi_Krener}, where design of control laws in the presence of center manifolds is presented, also uses the same nonsmooth Lyapunov function.
\end{remark}
 In the rest of this work, for ease of notation, we denote by $\dot{V}$, the derivative of $V$ when it exists and $\zeta^{\top}(\dot{y};\dot{w}), \zeta\in\partial V$, when the derivative does not exist. When we arrive at inequalities of the form $\dot{V} \leq f_r(y,w)$, it is made sure that such inequalities hold in both the sets where the function is differentiable and non-differentiable.    
\section{Event-triggered control} \label{event_4}
In this section, we use the LISS characterization of nonlinear systems with center manifolds presented in the previous section, to propose event-triggered control implementations. We begin by assessing the applicability of methods that already exist in literature. 
The LISS characterization 
\begin{align}
\dot{V}\leq-\alpha_D(\norm{(y;w)})+\beta_G(\norm{e})
\label{ISS_char_1}
\end{align} 
where, 
$\alpha_D$ and $\beta_G$ are class-${\cal K}$ functions was derived in the previous section. In \cite{tab1} and \cite{Postoyan}, a relative threshold based event-triggered control was proposed, where the events are triggered and the control is updated when the condition 
\begin{align}
\beta_G(\norm{e})\geq \sigma\alpha_D(\norm{(y;w)}), \;\; \sigma \in (0,1)
\label{trig_4}
\end{align}
is satisfied at event times $t_i, i \geq 0$. The input $u(t)$ under this implementation evolves as 
\begin{equation}
u(t) =\begin{cases} K(y(t_i);z(t_i)) & \text{if\;\;} \beta_G(\norm{e})< \sigma\alpha_D(\norm{(y;w)}) \\
K(y(t);z(t)) & \text{if\;\;} \beta_G(\norm{e})\geq\sigma\alpha_D(\norm{(y;w)}).
\end{cases}
\label{event_41}
\end{equation}
Note that the state $x$ of system \eqref{forced} is measured, from which $(y;z)$ is obtained through a linear transformation and the control $K(y;z)$ (design of which has been discussed in Section \ref{choice_of_a_controller}) is computed. From \eqref{ISS_char_1} and \eqref{event_41}, we have
\begin{align*}
\dot{V}\leq-(1-\sigma)\alpha_D(\norm{(y;w)}) < 0, \;\; \forall \;(y;w) \neq 0 
\end{align*}
%
and local asymptotic stability of the origin of system \eqref{final} is guaranteed. Although the thresholding condition \eqref{trig_4} is easily and directly derived from the LISS characterization, we see two major hurdles on closer observation.
\begin{enumerate}
    \item The triggering rule \eqref{trig_4} can be accurately checked only when $w=v-h(y)$ can be exactly computed. The center manifold $h(y)$ is determined by solving \eqref{partial} and there are systems for which $h(y)$ can be exactly computed (one such system is considered in Example 1). However, for most systems, only an approximation of $h(y)$ can be found.
    \item The non-existence of Zeno behaviour must be ensured. This involves showing that the inter-event times $t_{i+1}-t_i$ are lower bounded by a positive constant for all $i \geq 0$. Sufficient conditions that rule-out Zeno behaviour have been proposed in \cite{tab1, Postoyan} and these conditions require that the comparison functions $\alpha_D$ and $\beta_G$ in \eqref{ISS_char_1} are such that $\alpha_D^{-1} \circ \beta_G$ is Lipschitz continuous over compact sets. This assumption on the comparison functions has since been made in \cite{wang_lemmon, Girard_dynamic, NOZARI}, among many others. In our case, this assumption holds only when $\alpha_4 \in {\cal O}(\norm{y}^p)$, $p\leq1$. When $\alpha_4 \in {\cal O}(\norm{y}^p)$, $p>1$, the sufficient conditions from \cite{tab1, Postoyan} are not satisfied and no conclusion can be drawn regarding the existence or non-existence of Zeno behaviour under the implementation \eqref{event_41}.   
\end{enumerate}
In this section, we look to overcome these hurdles for the systems under consideration. To begin with, we focus on systems for which the center manifold can be exactly computed and look to overcome the second difficulty by proposing triggering conditions which are different from \eqref{trig_4}. Then, we turn our attention on systems for which the center manifold can only be approximately computed and propose modified triggering conditions.  

We begin by considering a class of nonlinear systems with center manifolds for which, in the Lyapunov characterization \eqref{ISS_char_1}, $\alpha_4 \in {\cal O}(\norm{y}^p)$, $p \leq 1$. As mentioned before, for such systems, the comparison functions in the LISS characterization \eqref{ISS_char_1} satisfy the sufficient conditions from \cite{tab1, Postoyan} and the triggering condition \eqref{trig_4} can be used for Zeno-free event-triggered implementation. However, the function $\beta_G$ in \eqref{ISS_char_1} is dependent on the constants $k_v$, $k_1$ and $k_2$, which characterize the local region where \eqref{ISS_char_1} holds and must be known to check the triggering condition \eqref{trig_4}. The triggering conditions proposed in this work possess a simpler relative thresholding structure and are easier to check than \eqref{trig_4}.  In the following Proposition, we present the triggering condition
\begin{equation}
\begin{aligned}
    \norm{e}\geq&\sigma\norm{(y;w)}, \\
    0< \sigma &\leq \frac{\lambda_{min}(Q)}{4\norm{PBK_1}}\frac{\sqrt{\lambda_{min}(P)}}{\sqrt{\lambda_{max}(P)}}
    \label{trig_4_simple}
\end{aligned}
\end{equation} 
which is computationally a simpler condition to check than \eqref{trig_4}.
\begin{proposition}
\label{first_new_proposition}
Consider the system \eqref{final_tranf}. Under the assumptions of Theorem \ref{center_manifold_theorem}, if the origin $y=0$ of the reduced system \eqref{dynamics_center} is locally asymptotically stable and there exists a LISS Lyapunov function such that in \eqref{ISS_char_1}, $\alpha_4 \in {\cal O}(\norm{y}^p)$, $p \leq 1$, then the origin of the overall system \eqref{final} is locally asymptotically stable under the event-triggered implementation with the triggering condition 
\eqref{trig_4_simple}. Moreover, the inter-execution times $(t_{i+1}-t_i)$ are lower bounded by a positive constant for all $i \geq 0$.
\end{proposition}
\begin{proof}
Consider the LISS characterization \eqref{LISS_Char} of system \eqref{final}.
\begin{align*}
\dot{V}\leq&-\alpha_4(\norm{y})-(1-s_f)\frac{\lambda_{min}(Q)}{2\sqrt{\lambda_{max}(P)}}\norm{w}\\
&+\left(k_vk_1+k_2\frac{\lambda_{max}(P)}{\sqrt{\lambda_{min}(P)}} -\frac{s_f \lambda_{min}(Q)}{2\sqrt{\lambda_{max}(P)}}\right)\norm{w}\\
&\hspace{0.15cm}+\left(k_vk_1+k_2\frac{\lambda_{max}(P)}{\sqrt{\lambda_{min}(P)}}+\frac{\norm{PB_2K_1}}{\sqrt{\lambda_{min}(P)}}\right)\norm{e}.
\end{align*}
As $\alpha_4 \in {\cal O}(\norm{y}^p)$, $p \leq 1$, there exists a neighbourhood of the origin in which $\alpha_4(\norm{y}) \geq l_1\norm{y}$ or $-\alpha_4(\norm{y}) \leq -l_1\norm{y}$, for some $l_1>0$. Using the notation $m_{p_1}=\frac{\norm{PB_2K_1}}{\sqrt{\lambda_{min}(P)}}$ and $\bar{m}_{p_1}=\left(k_vk_1+k_2\frac{\lambda_{max}(P)}{\sqrt{\lambda_{min}(P)}}+\frac{\norm{PB_2K_1}}{\sqrt{\lambda_{min}(P)}}\right)$ and with the event-triggering mechanism ensuring $\norm{e}\leq \sigma\norm{(y;w)}\leq \sigma (\norm{y}+\norm{w})$, we arrive at 
\begin{align}
\dot{V}&\leq-(l_1-\bar{m}_{p_1}\sigma)\norm{y}-\left(\frac{\lambda_{min}(Q)}{4\sqrt{\lambda_{max}(P)}}-m_{p_1}\sigma\right)\norm{w}\nonumber\\
&\hspace{0.05cm}+\left((1+\sigma)(k_vk_1+\frac{k_2\lambda_{max}(P)}{\sqrt{\lambda_{min}(P)}})-\frac{\lambda_{min}(Q)}{4\sqrt{\lambda_{max}(P)}}\right)\norm{w}.
\end{align}
When $\sigma$ is chosen as in \eqref{trig_4_simple} and $\delta_{yw}$ defining the set $S$ in \eqref{delta} is chosen such that the constants $k_1$ and $k_2$ satisfy
\begin{align*}
\left((1+\sigma)\left(k_vk_1+k_2\frac{\lambda_{max}(P)}{\sqrt{\lambda_{min}(P)}}\right)-\frac{\lambda_{min}(Q)}{4\sqrt{\lambda_{max}(P)}}\right) \leq 0 \\
(l_1-\bar{m}_{p_1}\sigma)>0,
\end{align*} 
we have
\begin{align*}
\dot{V}&\leq-(l_1-\bar{m}_{p_1}\sigma)\norm{y}-(\frac{\lambda_{min}(Q)}{4\sqrt{\lambda_{max}(P)}}-m_{p_1}\sigma)\norm{w} < 0,
\end{align*}
which implies that local asymptotic stability of the origin $(y;w)=0$ of system \eqref{final} is guaranteed. The existence of a positive lower bound for the inter-execution times $(t_{i+1}-t_i)$ for all $i\geq 0$ can be established as in the proof presented in \cite[Theorem III.1]{tab1}, where the evolution of the fraction $\frac{\norm{e(t)}}{\norm{(y(t);w(t))}}$ is analysed to find a lower bound $\tau$ for the time taken by the fraction to reach $\sigma$ from zero. For the triggering condition \eqref{trig_4_simple}, $\tau=\sigma/(L(1+\sigma))>0$, where $L=\text{min}(L_1,L_2)$ and the constants $L_1$ and $L_2$ are such that $\norm{\dot{y}}\leq L_1\norm{(y;w;e)}\leq L_1\norm{(y;w)}+L_1\norm{e}$ and $\norm{\dot{w}}\leq L_2\norm{(y;w;e)}\leq L_2\norm{(y;w)}+L_2\norm{e}$.
Therefore, the inter-execution times are uniformly lower bounded by a positive constant. 
\end{proof}
\begin{remark}
\label{remark_prop_2}
In showing the non-existence of Zeno behaviour for the triggering condition \eqref{trig_4}, the proof of \cite[Theorem III.1]{tab1} analyses an equivalent triggering condition $(\sigma\alpha_D)^{-1} \circ \beta_G (\norm{e}) \geq \norm{x}$. Under the assumption that the function $(\sigma\alpha_D)^{-1} \circ \beta_G$ is Lipschitz continuous, we have $(\sigma\alpha_D)^{-1} \circ \beta_G (\norm{e}) \leq P_l \norm{e} \leq \norm{x}$, where $P_l$ is the Lipschitz constant. The time taken by $\norm{e}/\norm{x}$ to reach $1/P_l$ from zero serves as an under-approximation for the time taken by $\alpha_D^{-1} \circ \beta_G (\norm{e})$ to reach $\norm{x}$. As the triggering condition in our case is chosen to be $\norm{e} \geq \sigma \norm{x}$ (with $x=(y;w)$), the proof of non-existence of Zeno behaviour
begins directly with the analysis of the evolution of the fraction $\norm{e}/\norm{(y;w)}$. 
\end{remark}
Next, we consider a class of systems for which, $\alpha_4 \in {\cal O}(\norm{y}^p)$, $p>1$ in the LISS characterization \eqref{ISS_char_1}. This case arises when the function $g_1$ in the reduced system \eqref{dynamics_center} has a polynomial approximation in a neighbourhood of the origin, 
that is, $\norm{g_1} \leq k_5\norm{y}^{p}$, $p>1$ in a neighbourhood of the origin and $V_1$ in \eqref{main_lyapunov_fun} is chosen as $\norm{y}$ or $\norm{y}^2$. For simplicity, in the rest of this work, we use $g_1$ to denote $g_1(y,h(y)+Ey,K(y;h(y)+Ey))$. 
\begin{assumption}
\label{second_assumption}
For the system \eqref{final_tranf}, the matrix $A_1=0$. In the dynamics of the reduced system \eqref{dynamics_center}, the function $g_1$ is such that $\norm{g_1} \leq k_5 \norm{y}^p$ and $y^{\top} g_1 \leq -k_6 \norm{y}^{p+1}$, with $p>1$ for some $k_5,k_6>0$ in a neighbourhood of the origin. The controller $u=K(y;z)=K_{11}z+K_{12}y+K_n(y)=K_1(y;v)+K_n(y)$ considered in subsection \ref{choice_of_a_controller} is such that $K_n(y)=0$ and the matrix $K_1=[K_{12} + E K_{11} \;\; K_{11}]=[0 \;\; K_{11}]$. 
\end{assumption}
As the focus of this work is on stabilization of the origin, we do not focus on systems for which the matrix $A_1$ has purely imaginary eigenvalues, leading to periodic orbits. Therefore, the matrix $A_1$ is assumed to be a zero matrix. Models and controllers of the two examples and the Mobile Inverted Pendulum robot considered in this work satisfy the conditions of Assumption \ref{second_assumption}. In the case of Examples 1 and 2, the controller $u=K_{11}z$ suffices to stabilize the dynamics on the center manifold and $K_{12}$ and $E$ are zero. In the case of the MIP robot, $K_{12}, E$ and $K_{11}$ are such that $K_{12} + E K_{11}=0$. The utility of the particular structure of the matrix $K_1$ is in showing both asymptotic stability of the origin and the non-existence on Zeno behaviour. This restriction on the matrix $K_1$ will be relaxed in Proposition \ref{third_new_proposition}.


The proof of the proposition presented next employs a non-smooth Lyapunov function. Specifically, a sum of two functions of the form $\sqrt{x^{\top}Px}$. We now recall a result that presents the subdifferential of the function $\sqrt{x^{\top}Px}$. 
\begin{lemma}[{\cite[page 181]{convex_funda}}]\label{subdiff_quad}
 Let $P \in \mathbb{R}^{n_1 \times n_1}$ be a symmetric positive definite matrix and $P = M^{\top}DM$ be its eigen-decomposition, where $M\in \mathbb{R}^{n_1 \times n_1}$ is an orthonormal matrix and $D\in \mathbb{R}^{n_1 \times n_1}$ is a diagonal matrix. The subdifferential of the function $f_P=\sqrt{x^{\top}Px} : \mathbb{R}^{n_1} \rightarrow \bar{\mathbb{R}}_{+}$ at $x=0$ is $\partial f_P(0) = \{g \in \mathbb{R}^{n_1} : \norm{g^{\top}MD^{-\frac{1}{2}}} \leq 1\}$. 
\end{lemma}
The above result will be used in evaluating the subdifferential of a nonsmooth Lyapunov function in the proof of the following proposition.
\begin{proposition}
\label{second_new_proposition}
Consider the system \eqref{final_tranf}. Under the assumptions of Theorem \ref{center_manifold_theorem}, if the origin $y=0$ of the reduced system \eqref{dynamics_center} is locally asymptotically stable and the conditions in Assumption \ref{second_assumption} are satisfied,
then the origin of the overall system \eqref{final} is locally asymptotically stable under the event-triggering condition
\begin{equation}
\begin{aligned}
\norm{e_v} &\geq \sigma(\norm{w}+\norm{y}^{(p+1)}) \\
0<\sigma &\leq \frac{(1-s_f)\lambda_{min}(Q)}{2\norm{PBK_1}}\frac{\sqrt{\lambda_{min}(P)}}{\sqrt{\lambda_{max}(P)}} 
\label{sigma_sf}
\end{aligned}
\end{equation}
where $s_f \in (0,1)$. Moreover, the inter-execution times $(t_{i+1}-t_i)$ are lower bounded by a positive constant for all $i \geq 0$.
\end{proposition}
\begin{proof}
Consider the LISS Lyapunov function candidate $V:\mathbb{R}^n \rightarrow \bar{\mathbb{R}}_{+}$ 
\begin{align}
V=\norm{y}+\sqrt{w^{\top}Pw}.
\label{lyap_prop2}
\end{align}
The function $V$ is continuously differentiable everywhere except on the set $N_e=\{(y;w)\in \mathbb{R}^n : w=0 \text{\;or\;} y=0\}$.  The functions $\alpha_1(\norm{(y;w)})=\min(1,\sqrt{\lambda_{\text{min}}(P)})\norm{(y;w)}$ and $\alpha_2(\norm{(y;w)})=\sqrt{2}\max(1,\sqrt{\lambda_{\text{max}}(P)})\norm{(y;w)}$ are class-$\mathcal{K}$ functions satisfying 
\begin{align*}
\alpha_1(\norm{(y;w)})\leq V((y;w)) \leq\alpha_2(\norm{(y;w)}).
\end{align*}
Taking the time derivative of $V$ along the trajectories of the system \eqref{final} on the set $\mathbb{R}^n \setminus N_e$, we get
\begin{align*}
\dot{V}&=\frac{y^{\top}\dot{y}}{\norm{y}}+\frac{1}{2\sqrt{w^{\top}Pw}}\left(\dot{w}^{\top}Pw+w^{\top}P\dot{w}\right).
\end{align*}
By Assumption \ref{second_assumption}, $A_1=0$, the function $g_1$ is such that $y^{\top} g_1 \leq -k_6 \norm{y}^{p+1}$, the control $u=K_{11}(v+e_v)$, $B_2K_{1}e=B_2K_{11}e_v$ in \eqref{final} and the functions $N_1$ and $N_2$ are functions of $y,w$ and $e_v$. This leads us to
\begin{align*}
\dot{V}&\leq -k_6 \norm{y}^{p} + \frac{y^{\top}N_1(y,w,e_v)}{\norm{y}}\\
&\hspace{0.33cm}+\frac{1}{2\sqrt{w^{\top}Pw}}\left((A_Kw+B_2K_{11}e_v+N_2(y,w,e_v))^{\top}Pw \right. \\
&\hspace{0.35cm} \left. +w^{\top}P(A_Kw+B_2K_{11}e_v+N_2(y,w,e_v))\right)\\
&\leq-k_6 \norm{y}^{p}+\norm{N_1(y,w,e_v)}-\frac{w^{\top}Qw}{2\sqrt{w^{\top}Pw}}\\
&\hspace{0.5cm} +\frac{1}{\sqrt{w^{\top}Pw}}(w^{\top}PB_2K_{11}e_v+w^{\top}PN_2(y,w,e_v))\\
&\leq-k_6 \norm{y}^{p}-\frac{\lambda_{min}(Q)}{2\sqrt{\lambda_{max}(P)}}\norm{w}+k_1(\norm{w}+\norm{e_v})\\
&\hspace{0.5cm} +\frac{\norm{PB_2K_{11}}}{\sqrt{\lambda_{min}(P)}}\norm{e_v}+\frac{k_2\lambda_{max}(P)}{\sqrt{\lambda_{min}(P)}}(\norm{e_v}+\norm{w}).
\end{align*}
With $s_f \in (0,1)$,
\beq
\begin{aligned}
\dot{V}&\leq-k_6 \norm{y}^{p}-(1-s_f)\frac{\lambda_{min}(Q)}{2\sqrt{\lambda_{max}(P)}}\norm{w}\\
&\hspace{0.3cm}+\left(k_1+k_2\frac{\lambda_{max}(P)}{\sqrt{\lambda_{min}(P)}}-s_f\frac{\lambda_{min}(Q)}{2\sqrt{\lambda_{max}(P)}}\right)\norm{w}\\
&\hspace{1cm}+\left(k_1+k_2\frac{\lambda_{max}(P)}{\sqrt{\lambda_{min}(P)}}+\frac{\norm{PB_2K_{11}}}{\sqrt{\lambda_{min}(P)}}\right)\norm{e_v}.
\end{aligned}
\label{comp_3}
\eeq
Using the notation $m_{p_2}=\frac{\norm{PB_2K_{11}}}{\sqrt{\lambda_{min}(P)}}$ and $\bar{m}_{p_2}=\left(k_1+k_2\frac{\lambda_{max}(P)}{\sqrt{\lambda_{min}(P)}}+\frac{\norm{PB_2K_{11}}}{\sqrt{\lambda_{min}(P)}}\right)$ and with the relative event-triggering rule ensuring $\norm{e_v}\leq \sigma(\norm{w}+\norm{y}^{p+1})$, we arrive at 
\begin{equation}
\begin{aligned}
\dot{V}\leq&-k_6 \norm{y}^{p}+\bar{m}_{p_2}\sigma \norm{y}^{p+1}-(1-s_f)\frac{\lambda_{min}(Q)}{2\sqrt{\lambda_{max}(P)}}\norm{w}\\
&+m_{p_2}\sigma\norm{w}+\left((1+\sigma)\left(k_1+k_2\frac{\lambda_{max}(P)}{\sqrt{\lambda_{min}(P)}}\right) \right.\\
&\left. -s_f\frac{\lambda_{min}(Q)}{2\sqrt{\lambda_{max}(P)}}\right)\norm{w}.
\end{aligned}
\end{equation}
With $s_y \in(0,1)$
\begin{align*}
\dot{V}\leq&-k_6(1-s_y) \norm{y}^{p}+(\bar{m}_{p_2}\sigma \norm{y}^{p+1}-k_6 s_y \norm{y}^{p})\\
&-(1-s_f)\frac{\lambda_{min}(Q)}{2\sqrt{\lambda_{max}(P)}}\norm{w}+m_{p_2}\sigma\norm{w} \nonumber\\
&+\left((1+\sigma)\left(k_1+k_2\frac{\lambda_{max}(P)}{\sqrt{\lambda_{min}(P)}}\right)\right.\\
&\left. -s_f\frac{\lambda_{min}(Q)}{2\sqrt{\lambda_{max}(P)}}\right)\norm{w}.
\end{align*}
With $\sigma$ chosen as in \eqref{sigma_sf},
in a small neighbourhood $S_1=\{(y,w) : \norm{(y;w)} < \delta_{yw}\}$, 
the constants $k_1$, $k_2$ from \eqref{notice_difference} can be chosen such that
\begin{align}
&\left((1+\sigma)\left(k_1+k_2\frac{\lambda_{max}(P)}{\sqrt{\lambda_{min}(P)}}\right) \right. \nonumber \\ & \left. \hspace{3cm} -s_f\frac{\lambda_{min}(Q)}{2\sqrt{\lambda_{max}(P)}}\right)\norm{w} \leq 0
\label{ensure}
\end{align} 
and 
\begin{align}
(\bar{m}_{p_2}\sigma \norm{y}^{p+1}-k_6 s_y \norm{y}^{p}) \leq 0.
\label{ensure_2}
\end{align}
For $(y;w) \in S_1$, we have
\begin{equation}
\begin{aligned}
\dot{V}&\leq-k_6(1-s_y) \norm{y}^{p}-(1-s_f)\frac{\lambda_{min}(Q)}{2\sqrt{\lambda_{max}(P)}}\norm{w} \\
&= - w_s((y;w)) < 0,
\label{vdot_or}
\end{aligned}
\end{equation}
where $w_s((y;w))$ is a positive definite function of $(y;w)$. It can be verified that, on the set $N_e$ where the Lyapunov function \eqref{lyap_prop2} is non-differentiable, the inequality \eqref{vdot_or} holds, that is,  $\zeta^{\top}(\dot{y};\dot{w}) \leq - w_s((y;w))$, for all $\zeta \in \partial V$ and $\partial V$ is found through Lemma \ref{subdiff_quad}. Although \eqref{vdot_or} holds in the set $S_1$, $S_1$ is not positively invariant as it is not a sub-level set of $V$ (which is positively invariant). The largest, connected sub-level set of $V$ contained in $S_1$ is $S_v=\{(y;w) \in S_1 \;|\; V((y;w))\leq \alpha_1(\delta_{yw})\}$. By \eqref{vdot_or}, asymptotic stability of the origin is guaranteed for all $x(0) \in S_v$.

Next, we prove the existence of a uniform positive lower bound for the inter-event times, ($t_{i+1}-t_i$),  when the system is initialized in the positively invariant set $S_v$. The error $e_v(t)=v(t_i)-v(t)=(w(t_i)+h(y(t_i)))-(w(t)+h(y(t)))=e_w(t)+e_h(t)$. Consider
\begin{align}
&\frac{d}{dt}\left(\frac{\norm{e_v}}{\norm{w}+\norm{y}^{p+1}}\right)  \nonumber\\
&= \frac{e_v^{\top}\dot{e}_v}{\left(\norm{w}+\norm{y}^{p+1}\right)\norm{e_v}}-\frac{\left(\frac{w^{\top}\dot{w}}{\norm{w}}+\frac{(p+1)\norm{y}^p y^{\top}\dot{y}}{\norm{y}}\right)\norm{e_v}}{\left(\norm{w}+\norm{y}^{p+1}\right)^2} \nonumber\\
&\leq \frac{\norm{\dot{e}_v}}{{\left(\norm{w}+\norm{y}^{p+1}\right)}}+\frac{\left(\norm{\dot{w}}+(p+1)\norm{y}^p\norm{\dot{y}}\right)\norm{e_v}}{{\left(\norm{w}+\norm{y}^{p+1}\right)^2}} \nonumber\\
&\leq \frac{\norm{\dot{e}_w}+\norm{\dot{e}_h}}{{\left(\norm{w}+\norm{y}^{p+1}\right)}}+\frac{\left(\norm{\dot{w}}+(p+1)\norm{y}^p\norm{\dot{y}}\right)\norm{e_v}}{{\left(\norm{w}+\norm{y}^{p+1}\right)^2}}.
\label{ineq_prop2}
\end{align}
From Assumption \ref{second_assumption}, $\norm{g_1} \leq k_5 \norm{y}^p$ for some $k_5>0$. As $\norm{h(y)} \in {\cal O}(\norm{y}^2)$ and $\norm{\frac{\partial h(y)}{\partial y}} \in {\cal O}(\norm{y})$, there exist constants $k_7>0$ and $k_8>0$ such that $\norm{h(y)}\leq k_7\norm{y}^2$ and $\norm{\frac{\partial h(y)}{\partial y}} \leq k_8\norm{y}$ in a small neighbourhood of $y=0$. From equations \eqref{final} and \eqref{notice_difference}, we have $\norm{\dot{y}} \leq k_5\norm{y}^p+k_1\norm{w}+k_1\norm{e_v}$, $\norm{\dot{w}}\leq(\norm{A_c}+k_2)\norm{w}+(\norm{B_2K_{11}}+k_2)\norm{e_v}$ and $\norm{\dot{e}_h}=\norm{\dot{h}(y)}\leq\norm{\frac{\partial h}{\partial y}}\norm{\dot{y}} \leq k_8k_5\norm{y}^{p+1}+k_8k_1\delta_{yw}\norm{w}+\delta_{yw} k_1k_8\norm{e_v}$ ($\norm{y}\leq\delta_{yw}$ has been used as $(y;w)\in S_1$). Consider the numerator of the first term in the right-hand side of \eqref{ineq_prop2}. Using the inequalities derived so far, we have
\begin{align}
\norm{\dot{e}_w}&+\norm{\dot{e}_h} \leq a_1 \left(\norm{w}+\norm{y}^{p+1}\right) + a_2 \norm{e_v}
\label{ineq1}
\end{align}
where, $a_1=\text{max}(\norm{A_c}+k_2+k_8k_1\delta_{yw},k_8k_5)$ and $a_2=\norm{B_2K_{11}}+k_2+\delta_{yw}k_1k_8$. For the numerator of the second term in \eqref{ineq_prop2}, we have
\begin{equation}
\begin{aligned}
\norm{\dot{w}}+(p+1)&\norm{y}^p\norm{\dot{y}} \\
& \leq a_3 \left(\norm{w}+\norm{y}^{p+1}\right)\norm{e} + a_4 \norm{e_v}^2 \label{ineq2}
\end{aligned}
\end{equation}
where $a_3=\text{max}(\norm{A_c}+k_2+(p+1)\delta_{yw}^pk_1,\delta_{yw}^{(p-1)}k_5)$ and $a_4=\norm{B_2K_{11}}+k_2+(p+1)\delta_{yw}^pk_1$. From \eqref{ineq1} and \eqref{ineq2}, we have 
\begin{equation}
\begin{aligned}
&\frac{d}{dt}\left(\frac{\norm{e_v}}{\norm{w}+\norm{y}^{p+1}}\right) \\
& \leq a_1 + \left(\frac{(a_2+a_3)\norm{e_v}}{\norm{w}+\norm{y}^{p+1}}\right) + a_4\left(\frac{\norm{e_v}}{\norm{w}+\norm{y}^{p+1}}\right)^2.
\end{aligned}
\label{comparison_1}
\end{equation}
Denoting $\norm{e_v}/(\norm{w}+\norm{y}^{p+1})$ by $e_s$, we have $\dot{e}_s \leq a_1 + (a_2+a_3) e_s + a_4 e_s^2$. Using the Comparison Lemma \cite{khalil}, it follows that $e_s(t) \leq \phi(t)$, where $\phi(t)$ is the solution of $\dot{\phi}=a_1 + (a_2+a_3) \phi + a_4 \phi^2$, initialized at $\phi(0)=0$. When an event occurs (when $e_s$ rises from zero to meet $\sigma$), the control is updated and $e_s(t)$ is reset to zero. Let $\tau_1$ be the time taken by $\phi(t)$ to evolve from $0$ to $\sigma$. As $e_s(t) \leq \phi(t)$, the time taken by $e_s(t)$ to reach $\sigma$ is greater than $\tau_1$. By the comparison lemma, 
\begin{align*}
    e_s(t) \leq \phi(t) = b \tan\left(\frac{b}{2}(t+c)\right) -(a_2+a_3)/(2a_4).
\end{align*}
where $b = \sqrt{4a_1a_4-(a_2+a_3)^2}$ and $c=(2/b)\tan^{-1}((a_2+a_3)/b)$. $\phi(\tau_1)=\sigma$ implies 
\begin{equation}
\begin{aligned}
    \tau_1=& \frac{2}{b}\left(\tan^{-1}\left(\frac{2a_4\sigma+(a_2+a_3)}{b}\right)-\tan^{-1}\left(\frac{a_2+a_3}{b}\right)\right)\\
    &\hspace{5.5cm} > 0.
\end{aligned}
\label{tau1_e}
\end{equation}
Thus, for all initializations $(y(0);w(0)) \in S_v$, there exists a uniform positive lower bound for the inter-event times.
\end{proof}
%
%
\begin{remark}
\label{remark_prop_3}
A triggering condition of the form $\norm{e_v}\geq \sigma(\norm{w}+\norm{y}^{(p+p_e)})$, $p_e \geq 1$ can also be used to arrive at \eqref{vdot_or} to guarantee local asymptotic stability of the origin of \eqref{final}. Proposition \ref{second_new_proposition} presents the case with $p_e=1$, which leads to the largest domain of attraction among this family of triggering conditions, as \eqref{ensure_2} holds in the largest set when $p_e=1$. As $p_e \rightarrow \infty$, the triggering condition tends to $\norm{e_v}\leq \sigma \norm{w}$ and Zeno behaviour is observed in simulations using this triggering condition. 
\end{remark}
Next, we recall an important qualitative property of the trajectories of nonlinear systems with center manifolds, under the assumptions of Theorem 1. This property, captured in the following lemma, is used in analysing the asymptotic behaviour of the inter-event times.
\begin{lemma}[Exponential convergence of trajectories to the center manifold \cite{carr}]\label{expoential_decay}
Consider the system \eqref{final_tranf} under the assumptions of Theorem \ref{existence}. Let $(y(t);v(t))$ be a solution to \eqref{final_tranf} with $\norm{(y(0);v(0))}$ sufficiently small. Then, there exist positive constants $C_1$ and $\mu$ such that 
\begin{equation}
\norm{v(t)-h(y(t))}\leq C_1e^{-\mu t}\norm{v(0)-h(y(0))}
\end{equation}
for all $t\geq0$.
\end{lemma}
Note that $\norm{v(t)-h(y(t))}=\norm{w(t)}$, is the euclidean distance between the state at time $t$ and the center manifold. Lemma \ref{expoential_decay} states that this distance decays exponentially to zero. The state $y(t)$ on the other hand decays to the origin slowly (non-exponentially), governed by the dynamics \eqref{dynamics_center} (if the dynamics \eqref{dynamics_center} is asymptotically stable). 

\subsection{Asymptotic behaviour of the inter-event times}
By Lemma \ref{expoential_decay}, 
we have
\begin{align*}
\lim_{w(t)\rightarrow 0} \frac{\norm{e_v}}{\norm{w}+\norm{y}^{p+1}}\;\; \underrightarrow{\;\;\text{exponentially}\;\;}\;\; \frac{\norm{e_h}}{\norm{y}^{p+1}}.
\end{align*}
as $\norm{e_v}=\norm{e_w+e_h}$. The behaviour of the inter-event times with the thresholding rule $\norm{e}\geq \sigma(\norm{w}+\norm{y}^{(p+1)})$ exponentially tends to the behaviour of the inter-event times with the thresholding rule $\norm{e_h}\geq \sigma\norm{y}^{(p+1)}$.  Consider
\begin{align*}
\frac{d}{dt}\left(\frac{\norm{e_h}}{\norm{y}^{p+1}}\right) &= \frac{-e_h^{\top}\dot{h}(y)}{\norm{e_h}\;\norm{y}^{p+1}} - \frac{(p+1)\;\norm{y}^p\;(y^{\top}\dot{y})}{(\norm{y}^{p+1})^2}\frac{\norm{e_h}}{\norm{y}} \\
&\leq \frac{\norm{\frac{\partial h(y)}{\partial y}}\;\norm{\dot{y}}}{\norm{y}^{p+1}} + \frac{(p+1)\;\norm{y}^p\;\norm{\dot{y}}\;\norm{e_h}}{(\norm{y}^{p+1})^2}.
\end{align*}
Using the inequalities leading up to \eqref{ineq1} and \eqref{ineq2}, but with $w=0$, we arrive at
\begin{align*}
\frac{d}{dt} \left(\frac{\norm{e_h}}{\norm{y}^{p+1}}\right) & \leq b_1 + (b_2+b_3)\frac{\norm{e_h}}{\norm{y}^{p+1}} + b_4\left(\frac{\norm{e_h}}{\norm{y}^{p+1}}\right)^2
\end{align*}
where $b_1=k_8k_5, b_2=\delta_{yw}k_1k_8, b_3=(p+1)\delta_{yw}^{p-1}k_5$ and $b_4=(p+1)\delta_{yw}^pk_1$. Following along the steps leading to \eqref{tau1_e}, we arrive at an estimate $\tau_2$ with the same form as \eqref{tau1_e} but with $a_i$ replaced by $b_i$. 
The estimate for the uniform lower bound on the inter-event times using $b_i$  is less conservative, as the inequalities leading up to \eqref{ineq1} and \eqref{ineq2} are now less conservative due to Lemma \ref{expoential_decay}.

\subsection{Triggering rule guaranteeing local ultimate boundedness of the trajectories of the closed-loop system}

In Assumption \ref{second_assumption}, the restriction on the matrix $K_1$ was needed to ensure both asymptotic stability and non-existence of Zeno behaviour. We now relax this assumption to consider a more general class of systems and show that under the implementation rule
\begin{equation}
\begin{aligned}
u(t)&= \begin{cases} 
0 & \text{if\;\;} t_1 \neq t_0, \;\;\forall \; t\in[t_0,t_1)\\
K(y(t_i);z(t_i)) & \text{if\;\;} t\in [t_i,t_{i+1}), i\geq 1
\end{cases}\\
t_1&= \text{min}\{t \geq t_0 \; | \; (y(t);w(t)) \in S_v \setminus S_2 \}\\
t_{i+1}&=\min\{ t \geq t_i \;|\; \norm{e}\geq \sigma(\norm{w}+\norm{y}^{(p+1)}) \\
    &\hspace{2.3cm} \text{\;and\;} (y(t);w(t)) \in S_v \setminus S_2 \}, i 
    \geq 1
\label{imp2}    
\end{aligned}
\end{equation}
where
\begin{align}
S_2=\{(y;w)\; | \; \norm{(y;w)} < \alpha_2^{-1}\alpha_1(r_s)=r_1\}, 
\label{eq_ball}
\end{align}
the trajectories of the closed-loop system are ultimately bounded by a ball of desired radius $r_s$. If $(y(0);w(0)) \in S_v \setminus S_2$, the first event instant $t_1=t_0$ and the second case defining $u(t)$ is active for all $t \geq t_0$. If $(y(0);w(0)) \in S_2$, then $t_1 \neq t_0$ and the first case defining $u(t)$ is active for $t\in[t_0,t_1)$ before the second case takes over for all $t\geq t_1$. 

The following Assumption and Proposition capture this scenario. For simplicity, we use the notation $\mathcal{B}_r$ to denote a ball of radius $r$.
\begin{assumption}
\label{third_assumption}
For the system \eqref{final_tranf}, the matrix $A_1=0$. The function $g_1$ in the dynamics of the reduced system \eqref{dynamics_center} is such that $\norm{g_1} \leq k_5 \norm{y}^p$ and $y^{\top} g_1 \leq -k_6 \norm{y}^{p+1}$, with $p>1$ for some $k_5,k_6>0$ in a neighbourhood of the origin. 
\end{assumption}
\begin{proposition}
\label{third_new_proposition}
Consider the system \eqref{final_tranf}. Under the assumptions of Theorem \ref{center_manifold_theorem}, if the origin $y=0$ of the reduced system \eqref{dynamics_center} is locally asymptotically stable, the conditions in Assumption \ref{third_assumption} are satisfied and $\mathcal{B}_{r_s} \subset S_v$, 
then the trajectories of the system \eqref{final} are locally ultimately bounded by $\mathcal{B}_{r_s}$, under the event-triggered implementation \eqref{imp2}, with $\sigma$ chosen  according to \eqref{sigma_sf}. Moreover, the inter-execution times $(t_{i+1}-t_i)$ are lower bounded by a positive constant for all $i \geq 0$.
\end{proposition}
\begin{proof}
From \eqref{vdot_or}, \eqref{eq_ball} and the implementation rule \eqref{imp2}, we are lead to
\begin{align*}
\dot{V}\leq -w_s((y;w)), \;\; \forall \; \norm{(y;w)}\geq \alpha_2^{-1}\alpha_1(r_s).
\end{align*}
By \cite[Theorem 4.18]{khalil}, we conclude that the trajectories of the closed-loop system \eqref{final_tranf}-\eqref{imp2} are ultimately bounded by $\mathcal{B}_{r_s} \subset S_v$ and this ball is reached in finite time. 

Next we show that the inter-event times $(t_{i+1}-t_i)$ are uniformly lower bounded by a positive constant for all $(y(t_0);w(t_0)) \in S_v$. Under the implementation rule \eqref{imp2}, events are triggered only in $S_v \setminus S_2$ and the control is not updated when the trajectory enters $S_2$. When the system is initialized in $S_2$, the control is set to zero until the system leaves $S_2$.
For each $i$, $\norm{(y(t_i);w(t_i))}\geq r_1$ and $e(t_i)=0$. The next event occurs at $t_{i+1}$, when $\norm{e}$ rises from zero and meets $\sigma(\norm{w}+\norm{y}^{p+1})$ in $S_v\setminus S_2$. 
In $S_v\setminus S_2$
\begin{align*}
\sigma_l \triangleq \sigma (r_1+r_1^{p+1}) &\leq \sigma(\norm{w}+\norm{y}^{p+1}) \\
&\leq \sigma (\delta_{yw}+\delta_{yw}^{p+1}) \triangleq \sigma_u.
\end{align*}
Consider the evolution of $\norm{e}$ along the trajectories of the event-triggered closed-loop system, between two consecutive event instants
\begin{align*}
\frac{d\norm{e}}{dt} \leq \norm{\dot{e}}=\norm{-(\dot{y};\dot{w})}.
\end{align*}
The function $f(x,u)$ in \eqref{forced} is twice continuously differentiable and due to the sequence of smooth coordinate transformations relating $x$ and $(y;w)$, there exists a constant $L_1$ in $S_1$ such that $\norm{(\dot{y};\dot{w})}\leq L_1\norm{(y;w;e)}\leq L_1\norm{(y;w)}+L_1\norm{e}\leq L_1\delta_{yw}+L_1\norm{e}.$
Using the notation $\delta \triangleq L_1\delta_{yw}$, 
\begin{align*}
\frac{d\norm{e}}{dt}\leq \norm{-(\dot{y};\dot{w})}&\leq L_1 \delta_{yw}+L_1\norm{e}=L_1\norm{e}+\delta.
\end{align*} 
Using the Comparison Lemma \cite{khalil} with $e(0)=0$, we have
\begin{align}
\norm{e(t)}\leq \frac{\delta}{L_1}(e^{L_1t}-1).
\label{ltau3}
\end{align}
The time taken by the right-hand side of \eqref{ltau3} to rise from $0$ to $\sigma_l$ serves as a lower bound for the inter-event times. The lower bound $\tau_3$ is found by solving
\begin{align}
\frac{\delta}{L_1}(e^{L_1\tau_3}-1)= \sigma_l \implies \tau_3 = \frac{1}{L_1} \ln \left(1+\frac{\sigma_l L_1}{\delta}\right)>0.
\label{est2}
\end{align} 
Thus, we have shown that the inter-event times are uniformly lower bounded by $\tau_3>0$.
\end{proof}
The radius $r_s$ of the ultimate bound can be made arbitrarily small, but by allowing for small estimates of inter-event times, as $\tau_3$ in \eqref{est2} is a function of $\sigma_l$ which grows smaller as $S_2$ becomes smaller. 

Next, we take up an example where the center manifold is exactly computable and demonstrate the application of the triggering condition presented in Proposition \ref{second_new_proposition}.
\subsection{Example 1} \label{Example41}
Consider the system 
\begin{equation}
\begin{aligned}
\dot{y} &= -yv \\
\dot{v} &= v+u+y^2-2v^2.
\end{aligned}
\label{example_1_origin}
\end{equation}
The nonlinear function $g_1(y,v,u)=-yv$ is independent of the control input $u$. The control $u$ determines the structure of the center manifold and its stability indirectly via the $v$-subsystem. With the controller $u=-2v$, the closed-loop system is
\begin{align*}
\dot{y} &= -yv \\
\dot{v} &= -v+y^2-2v^2.
\end{align*}
The center manifold of this system can be computed exactly \cite{roberts_1985}. By solving the PDE \eqref{partial}, the center manifold is found to be $v=y^2$. The dynamics on the center manifold is
\begin{align}
\dot{y} &=g_1(y,h(y))= -y^3.
\label{example1_center}
\end{align}
We have $\norm{g_1}=|y|^3$ and $y^{\top}g_1\leq -|y|^4$ with $k_5, k_6=1$. Therefore, Assumption \ref{second_assumption} is satisfied. The local asymptotic stability of $y=0$ can be shown using the Lyapunov function $V_1=\frac{1}{2}y^2$. With the change of variables $w=v-y^2$, we obtain
\begin{equation}
\begin{aligned}
\dot{y} &= -y(w+y^2) \\
\dot{w} &= -w-2w(w+y^2).
\end{aligned}
\label{ref_ydot}
\end{equation}
Introducing the measurement error $e_v=v(t_i)-v(t)$, we obtain
\begin{equation}
\begin{aligned}
\dot{y} &= -y^3-yw \\
\dot{w} &= -w-2e_v-w(w+y^2).
\end{aligned}
\label{ex1er}
\end{equation}
Now $N_1=-yw$ and $N_2=-w(w+y^2)$. The functions $N_1$ and $N_2$ satisfy conditions \eqref{crucial}. As $N_1$ and $N_2$ are independent of $u$ and therefore $e_v$, the right-hand sides of the inequalities \eqref{notice_difference} are functions of only $w$. In the set $S_1=\{(y,w) \; | \; \norm{y}\leq\frac{1}{\sqrt{6}} \text{\;and\;} \norm{w}\leq \frac{1}{6}\}$,
we have $\norm{N_1}\leq\frac{1}{\sqrt{6}}\norm{w}  \text{\;and \;} \norm{N_2}\leq\frac{1}{6}\norm{w}$ with $k_1=\frac{1}{\sqrt{6}}$ and $k_2=\frac{1}{6}$.    
Using the LISS Lyapunov function $V=|y|+\sqrt{w^{\top}w}$, with $s_f=\frac{3}{4}$ and $s_y=\frac{1}{2}$ and the triggering condition $\norm{e_v}\geq \sigma(\norm{w}+\norm{y}^{4})$, $\sigma=1/16$ (according to \eqref{sigma_sf}), we have
\begin{align*}
\dot{V}\leq -\frac{1}{2}|y|^3-\frac{1}{16}\norm{w}
\end{align*}
 and through Proposition \ref{second_new_proposition}, local asymptotic stability of the origin of \eqref{ex1er} is guaranteed with Zeno-free event-triggering.
 


In Figure \ref{Cumulative_Example1}, simulation results of the event-triggered closed-loop system are presented. Trajectories from three initial conditions are plotted in Figure \ref{ex1p1} along with the center manifold $v=y^2$. The trajectories tend to the center manifold exponentially and the evolution along the center manifold is significantly slower in comparison. The mechanism of event-triggering is shown in Figure \ref{ex1p2}, by plotting the evolution of the error $\norm{e_v}$ and the threshold $\frac{1}{16}(\norm{w}+\norm{y}^{4})$. 

The evolution of inter-event times for three initial conditions are shown in Figure \ref{ex1p3}. The inter-event times are lower bounded by \unit{30.3}{\milli\second}. As $h(y)$ has been determined exactly, the function $g_1(y,h(y)+Ey,K(y;h(y)+Ey))$ is known in closed-form and good estimates of the constants $k_5, k_6, k_7, k_8$ and the neighbourhoods in which they are valid can be found. However, the estimates of minimum inter-event times from \eqref{tau1_e} and \eqref{est2} are conservative, as they depend on Lipschitz constants and bounds on polynomial functions. To get a better estimate, the event-triggered closed-loop system was simulated for ten initial conditions $\left(0.1\cos(\frac{2\pi k}{10}),0.1\sin(\frac{2\pi k}{10})\right)$, $k=0, 1, 2, \ldots, 9$ for \unit{25}{\second}. The trajectories practically converge to the center manifold by \unit{15}{\second}. The mean inter-event time before \unit{15}{\second} is found to be \unit{40.8}{\milli\second} and between \unit{15}{\second} and \unit{25}{\second} is found to be \unit{33}{\milli\second}. The minimum inter-event time in these simulations (MIETs) was found to be \unit{30.3}{\milli\second} and the average inter-event time was found to be \unit{36.3}{\milli\second}. To assess the performance of event-triggered control with respect to time-triggered control, we choose MIETs as the sampling time for time-triggered control. From Figure \ref{ex1p4}, we see that the performances of time-triggered and event-triggered control are a close match. However, the number of control updates is much higher in time-triggered control, thus making a case for the use of event-triggered control. 

\begin{figure} 
    \centering
  \subfloat[Phase portrait of the event-triggered closed-loop system\label{ex1p1}]{%
       \includegraphics[width=0.49\linewidth]{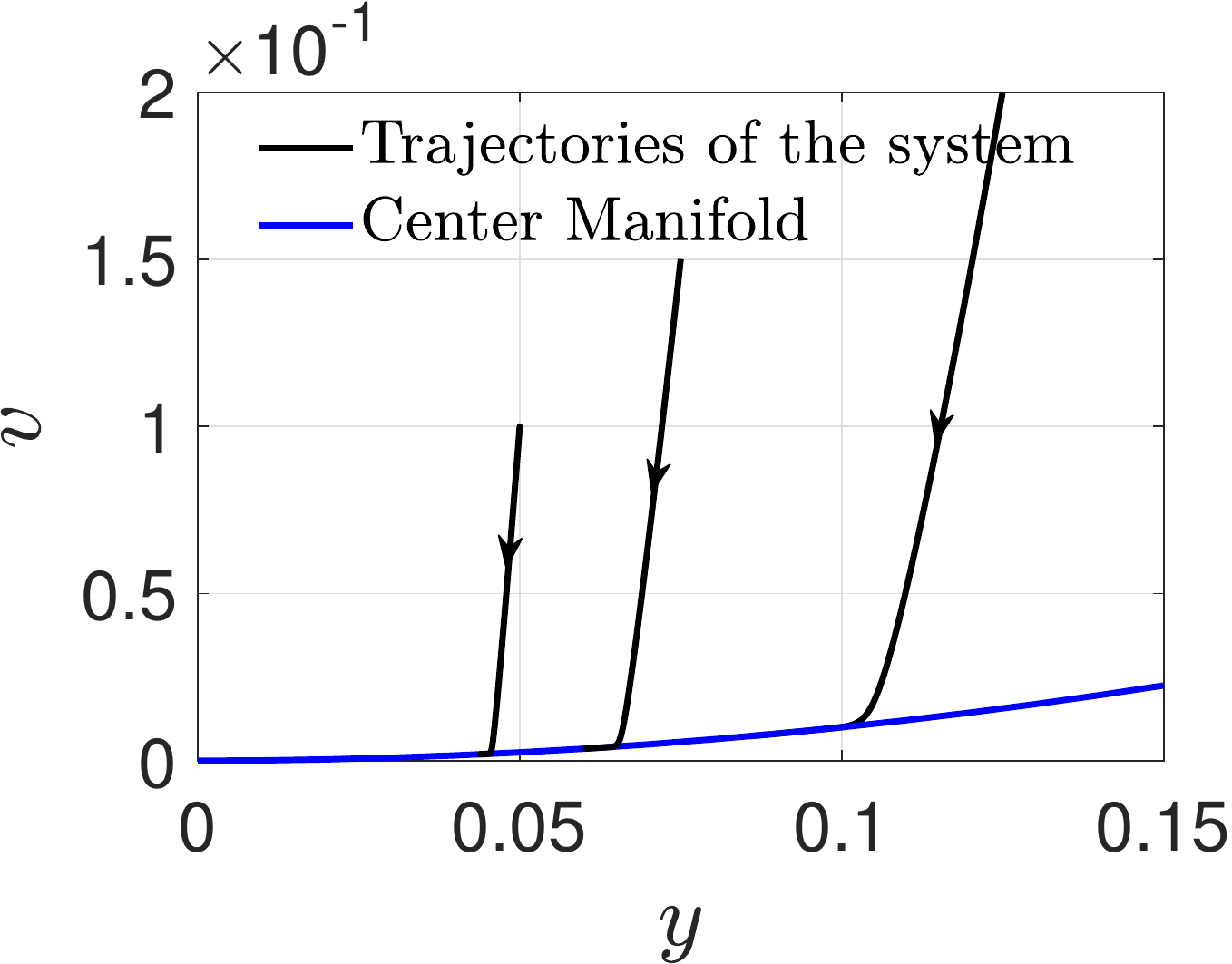}}
    \hfill
  \subfloat[Evolution of the error and the threshold\label{ex1p2}]{%
        \includegraphics[width=0.49\linewidth]{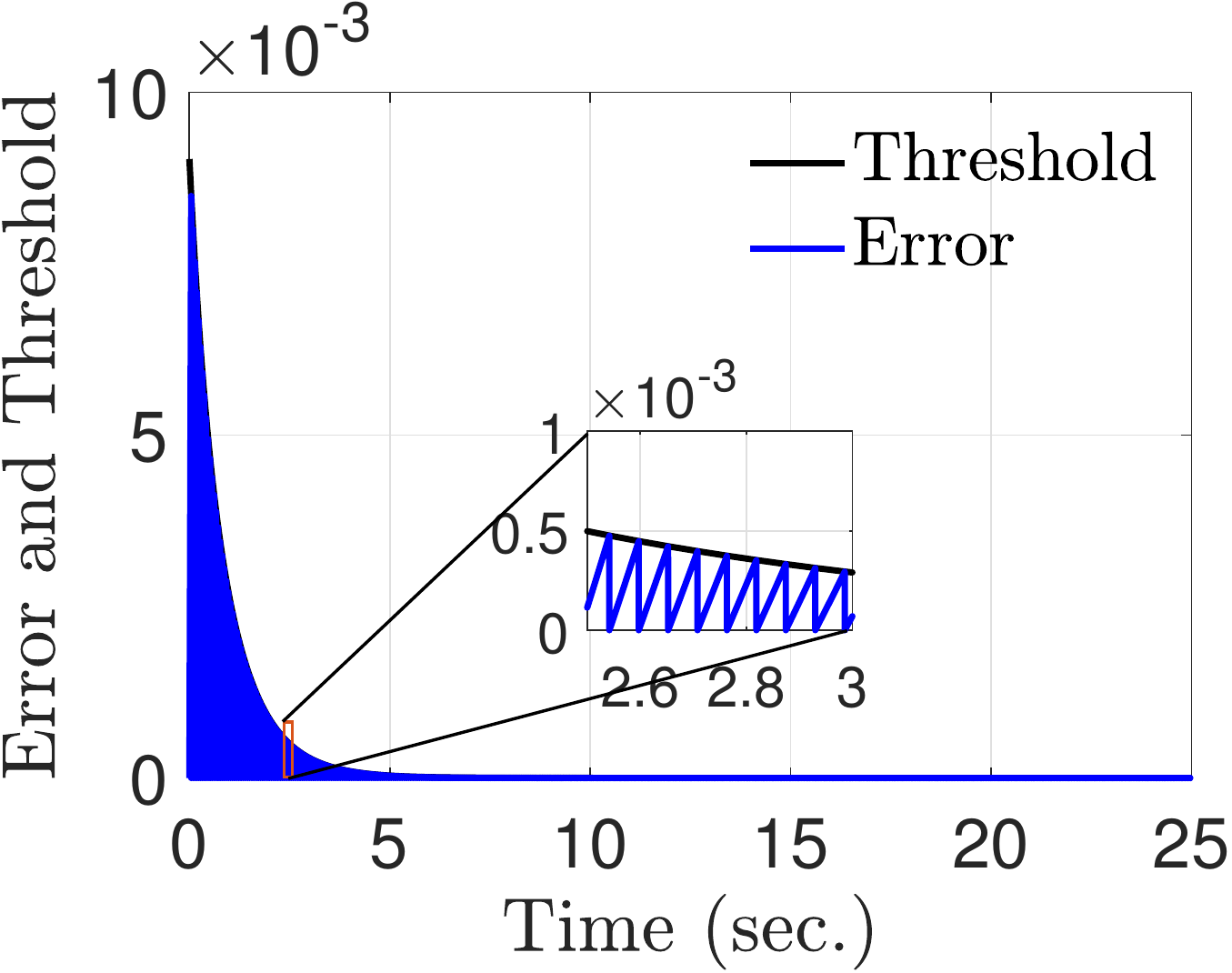}}
    \\ \vspace{0.25cm}
  \subfloat[Inter-event times\label{ex1p3}]{%
        \includegraphics[width=0.49\linewidth]{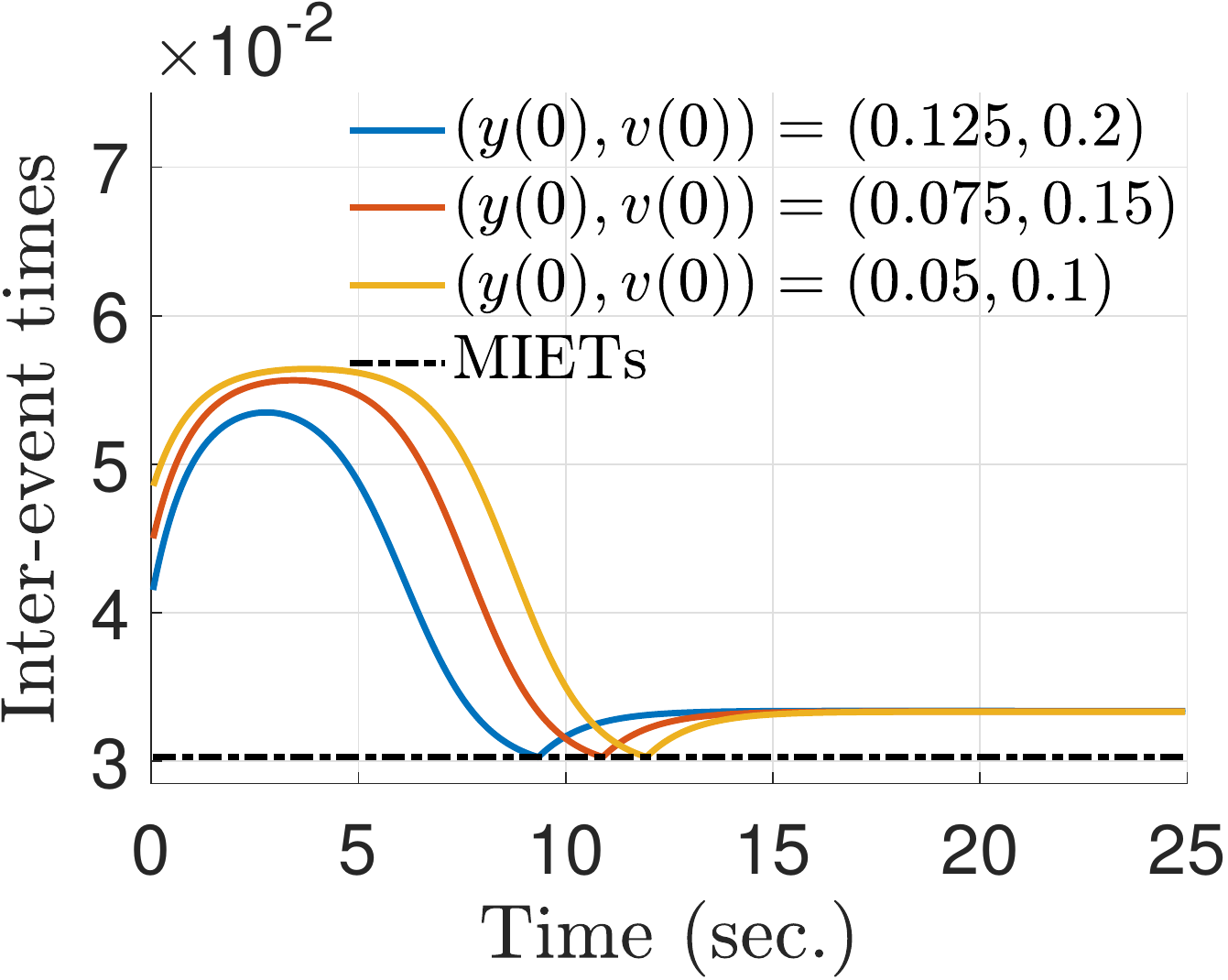}}
    \hfill
  \subfloat[Performance of time-triggered and event-triggered control\label{ex1p4}]{%
        \includegraphics[width=0.49\linewidth]{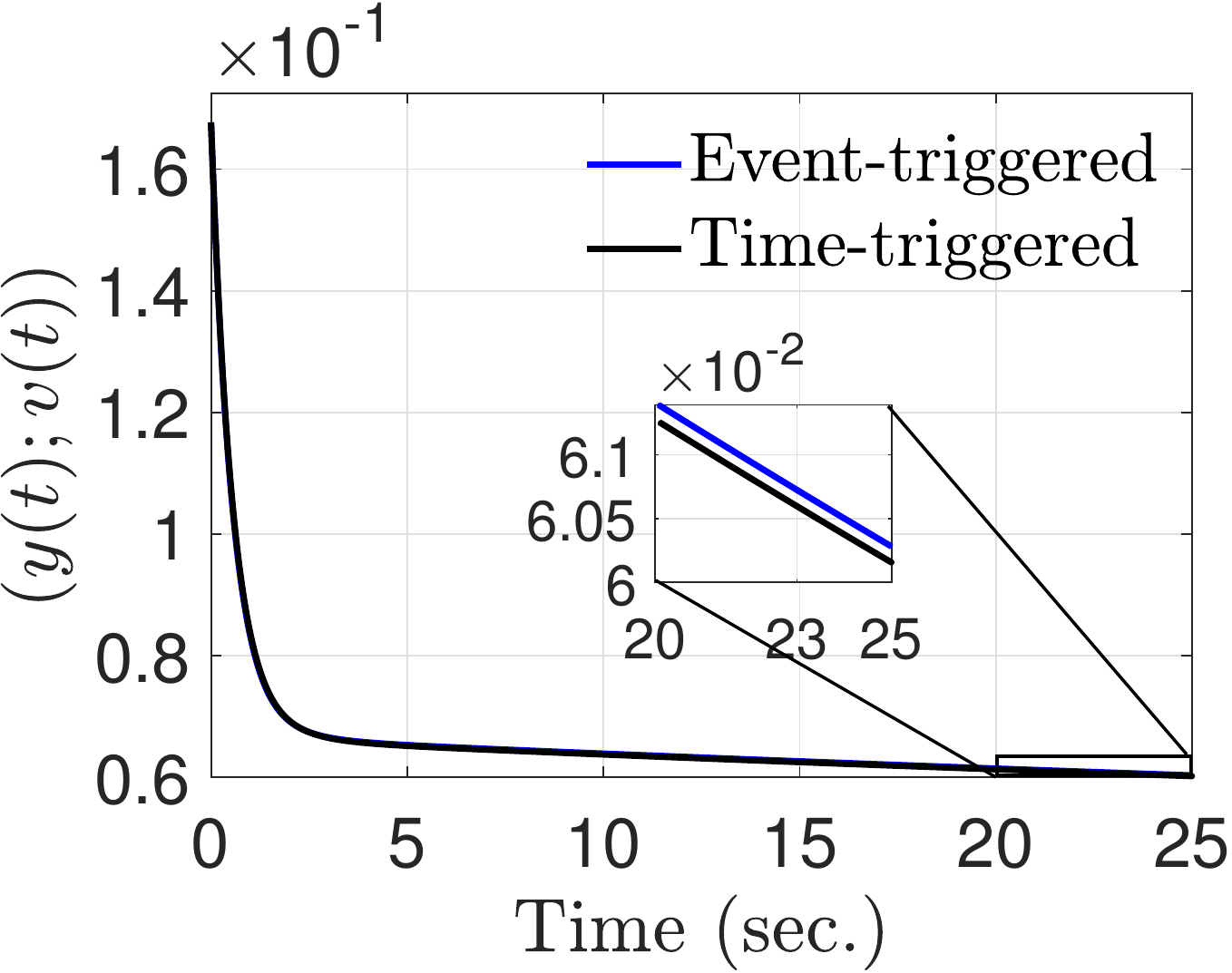}}
  \caption{Simulation results for event-triggered implementation of a controller designed for system \eqref{example_1_origin}, with the triggering rule $\norm{e_v}\geq\sigma(\norm{w}+y^4)$ for $\sigma=\frac{1}{16}$.}
  \label{Cumulative_Example1}
\end{figure}
The triggering conditions in Propositions \ref{first_new_proposition}-\ref{third_new_proposition} can be accurately checked only when the variable $w=v-h(y)$ is exactly computed. For many systems, $h(y)$ can only be found upto an approximation by solving \eqref{partial}.
This brings into question the utility of the change of coordinates from $(y;v)$ to $(y;w)$ and the following questions arise naturally - 1) Can triggering conditions be designed without resorting to the change of variables, directly in the $(y;v)$ coordinates? 2) Can we use the available partial knowledge of the center manifold to design triggering conditions? In the next two sections, we take a relook at the problem formulation, the assumptions and the proofs presented so far and address the questions posed above.  
%
%
\section{Need for the change of variables from $v$ to $w$}
\label{need}
The brief discussion following Lemma \ref{expoential_decay} in Section \ref{event_4} on the nature of the trajectories of nonlinear systems with center manifolds suggested that $(y;w)$ coordinates is the appropriate coordinate system to understand the behaviour of nonlinear systems with center manifolds. 
The change of variables from $v$ to $w=v-h(y)$ results in the functions $N_i$ satisfying conditions \eqref{crucial}, which are crucial for the proofs of Propositions \ref{main_result}-\ref{third_new_proposition}. On proceeding without the change of variables from \eqref{final_tranf}, we see that the functions
\begin{align*}
N_1&=g_1(y,v+Ey,K(y+e_y;v+e_v+E(y+e_y))\\
&\hspace{0.5cm}-g_1(y,h(y)+Ey,K(y;h(y)+Ey))\\
N_2&=g_2(y,v+Ey,K(y+e_y;v+e_v+E(y+e_v)))
\end{align*} 
satisfy
\begin{align}
N_i(0,0,0)=0, \frac{\partial N_i}{\partial y}(0,0,0)=0, \\ \frac{\partial N_i}{\partial v}(0,0,0)=0 \text{\; \;and\; \;} \frac{\partial N_i}{\partial e}(0,0,0)=0
\label{crucial_with_g_i}.
\end{align} 
In a small neighborhood of the origin
\begin{align}
S_3=\{(y;v) \;|\; \norm{(y;v)} < \delta_{yv}\}
\label{set_3}
\end{align}
we have for $i=1,2$,
\begin{align}
\norm{N_i}\leq k_i\norm{(y;v;e)}\leq k_i(\norm{y}+\norm{v}+\norm{e})
\label{newki}
\end{align}
where $k_i$ are positive constants which can be made arbitrarily small by decreasing $\delta_{yv}$. The difference between \eqref{newki} and inequalities \eqref{notice_difference} is the presence of the term $\norm{y}$. The absence of this term in \eqref{notice_difference} was crucial for all the proofs presented so far. Next, we analyse the local input-to-state stability of system \eqref{with_error} without the change of coordinates, by choosing the LISS Lyapunov function $V_4=V_1+\sqrt{v^{\top}Pv}$ (which is \eqref{main_lyapunov_fun} with $w$ replaced by $v$). Taking the time derivative of $V_4$ along the trajectories of system \eqref{with_error}, we have
\begin{align}
\dot{V}\leq&-(1-s_y)\alpha_4(\norm{y})-(1-s_f)\frac{\lambda_{min}(Q)}{2\sqrt{\lambda_{max}(P)}}\norm{v} \nonumber\\
&+\left(k_vk_1+k_2\frac{\lambda_{max}(P)}{\sqrt{\lambda_{min}(P)}}+\frac{\norm{PB_2K_1}}{\sqrt{\lambda_{min}(P)}}\right)\norm{e}\nonumber\\
&+\left(\left(k_vk_1+k_2\frac{\lambda_{max}(P)}{\sqrt{\lambda_{min}(P)}}\right)\norm{y}-s_y\alpha_3(\norm{y})\right) \nonumber\\
&+\left(k_vk_1+k_2\frac{\lambda_{max}(P)}{\sqrt{\lambda_{min}(P)}}-s_f\frac{\lambda_{min}(Q)}{2\sqrt{\lambda_{max}(P)}}\right)\norm{v}
\label{event_derivation}
\end{align}
where $s_y, s_f \in (0,1)$. When $\delta_{yv}$ defining the set $S_3$ is chosen such that $k_1$ and  $k_2$ in \eqref{newki} lead to
\begin{align}
\left(\left(k_vk_1+k_2\frac{\lambda_{max}(P)}{\sqrt{\lambda_{min}(P)}}\right)\norm{y}-s_y\alpha_4(\norm{y})\right)\leq 0
\label{asymptotic_condition}
\end{align}
and 
\begin{align}
\left(k_vk_1+k_2\frac{\lambda_{max}(P)}{\sqrt{\lambda_{min}(P)}}-s_f\frac{\lambda_{min}(Q)}{2\sqrt{\lambda_{max}(P)}}\right) \leq 0,
\label{condition_2}
\end{align} 
we obtain
\begin{equation}
\begin{aligned}
\dot{V}\leq&-(1-s_y)\alpha_3(\norm{y})-(1-s_f)\frac{\lambda_{min}(Q)}{2\sqrt{\lambda_{max}(P)}}\norm{v} \\
&+\left(k_vk_1+k_2\frac{\lambda_{max}(P)}{\sqrt{\lambda_{min}(P)}}+\frac{\norm{PB_2K_1}}{\sqrt{\lambda_{min}(P)}}\right)\norm{e}.
\end{aligned}
\label{iss_v}
\end{equation}
By making $S_3$ smaller, the constants $k_1$ and $k_2$ can be chosen such that \eqref{condition_2} is satisfied. Such a choice for the constants also leads to \eqref{asymptotic_condition} only if $\alpha_4(\norm{y}) \in {\cal O}(\norm{y}^p)$ with $p \leq 1$. In this case, the origin of system \eqref{final_tranf} is locally input-to-state stable by Definition \ref{Local_Input}. If $p>1$, \eqref{asymptotic_condition} is not satisfied close to the origin for any choice of $k_1$ and $k_2$ and \eqref{iss_v} holds only in the set $S_3\setminus S_4$ where 
\begin{equation}
\begin{aligned}
S_4=\{y \in \mathbb{R}^k\;|\;((k_vk_1+&k_2\frac{\lambda_{max}(P)}{\sqrt{\lambda_{min}(P)}})\norm{y} \\
&-s_y\alpha_4(\norm{y}))\geq 0\}.
\label{final_ball}
\end{aligned}
\end{equation}
Therefore, LISS of only the set $S_4$ can be concluded in this case. With the chosen LISS Lyapunov function in the $(y,v)$ coordinates, LISS of the origin of system \eqref{final_tranf} cannot be shown for all comparison function $\alpha_4$ given by the converse Lyapunov theorem. The function $\alpha_4$ guaranteed by the converse Lyapunov theorem (under the assumption that the reduced system is locally stable) can be any class-${\cal K}$ function and may or may not belong to $O(\norm{y}^p)$, $p \leq 1$. The change of variables from $v$ to $w$ helps in establishing LISS of the origin for all comparison functions $\alpha_4$.

The triggering conditions in Propositions \ref{first_new_proposition}-\ref{third_new_proposition} can be used by replacing $w$ by $v$, but Zeno-free local asymptotic stability of the origin is guaranteed when $\alpha_4 \in O(\norm{y}^p)$, $p \leq 1$ and only Zeno-free local asymptotic stability to the set \eqref{final_ball} can be guaranteed when $\alpha_4 \in O(\norm{y}^p)$, $p > 1$. 

The insights gained in the above discussion and in the discussion presented in the next section are formalized in Propositions \ref{first_corollary}-\ref{third_corollary} presented in section \ref{example_4}. Next, we consider the case of systems for which only approximate knowledge of the center manifold is available.  
\section{Event-triggered control with approximate knowledge of the center manifold} \label{example_4}
In this section, we analyse the applicability of Propositions \ref{second_new_proposition}-\ref{third_new_proposition} and the satisfaction of Assumptions \ref{second_assumption}, to see if triggering rules can be designed that require only partial knowledge of the center manifold. We begin with an illustrative example that brings out the differences in local stability analysis of systems for which only an approximation of the center manifold can be found and those for which the center manifold can be exactly found (Example 1).
\subsection {Example 2}
\noindent Consider the system
\begin{equation}
\begin{aligned}
\dot{y}&=-y(z_1-4z_2)\\
\left[\begin{matrix}
\dot{z}_1\\
\dot{z}_2
\end{matrix}\right]&=\left[\begin{matrix}0 & 1\\
-2 & 3\end{matrix}\right]\left[\begin{matrix}z_1\\
z_2\end{matrix}\right]+\left[\begin{matrix}0\\1\end{matrix}\right]u+\left[\begin{matrix}
y^2 \\0
\end{matrix}\right].
\end{aligned}
\label{sys_ex2}
\end{equation}
A linear state-feedback controller $u=[1 \;\; -4]z$ is used to place the poles of the linear part of the $z$-subsystem at $-0.5 \pm j0.0866$. The closed-loop system is
\begin{align}
\dot{y}&=-y(z_1-4z_2) \nonumber\\
\left[\begin{matrix}
\dot{z}_1\\
\dot{z}_2
\end{matrix}\right]&=\left[\begin{matrix}0 & 1\\
-1 & -1\end{matrix}\right]\left[\begin{matrix}z_1\\
z_2\end{matrix}\right]+\left[\begin{matrix}
y^2 \\0
\end{matrix}\right].
\label{example_closed}
\end{align}
The one-dimensional center manifold satisfying \eqref{partial}, found accurately up to order two, is given by
\begin{align*}
\left[\begin{matrix}
z_1 \\ z_2
\end{matrix}\right]=\left[\begin{matrix}
h_1(y) \\ h_2(y)
\end{matrix}\right]=\left[\begin{matrix}
y^2+{\cal O}(|y|^4) \\ -y^2+{\cal O}(|y|^4)
\end{matrix}\right].
\end{align*}
The dynamics on the center manifold is governed by
\begin{align}
\dot{y}=g_1(y,h(y))=-5y^3+{\cal O}(|y|^5).
\label{example_center}
\end{align}
Consider the nonsmooth Lyapunov function $V_1(y)=|y|$. On the set $\mathbb{R} \setminus \{0\}$,
\begin{align*}
\dot{V}_1=\frac{y\dot{y}}{|y|}=-5|y|^3 +{\cal O}(|y|^5).
\end{align*}
In a small neighbourhood of the origin 
\begin{equation}
\begin{aligned}
\dot{V}_1&\leq-5(1-\theta)|y|^3 - 5\theta |y|^3 + k_p|y|^5\\
&\leq -5(1-\theta)|y|^3<0 \;\;\;\;\; \forall\;|y|^2 \leq \dfrac{5\theta}{k_p}
\end{aligned}
\label{vdot_ex2}
\end{equation}
for some $k_p>0$ and local asymptotic stability of \eqref{example_center} is concluded. The partial knowledge of the center manifold suffices for the local stability analysis of this system. Through \eqref{example_center} and \eqref{vdot_ex2}, we see that Assumption \ref{second_assumption} is satisfied. However, since $g_1(y,h(y))$ in \eqref{example_center} cannot be obtained in closed form, only crude estimates of the neighbourhoods where the conditions in Assumption \ref{second_assumption} hold can be found. Similarly, good estimates of the neighborhoods where \eqref{crucial} are satisfied cannot be obtained. These estimates however, improve as the center manifold is computed to higher powers.  
\begin{remark}
If establishing the local stability of the equilibrium point is of sole interest, approximate knowledge of the center manifold suffices. In practice, as in \cite{tethered, vijay_TCST}, only the local analysis is performed and it is argued that the region of attraction is much larger than the estimates derived as above.
\end{remark}
Next, we analyse the consequences of using the triggering conditions of the form $\norm{e}\geq \sigma(\norm{w^{a}}+\norm{y})$ and $\norm{e_v}\geq \sigma(\norm{w^{a}}+\norm{y}^{(p+1)})$, which possess the same structure as the triggering conditions in Propositions \ref{first_new_proposition} and \ref{second_new_proposition} respectively, but with $w=v-h(y)$ replaced by the approximation $w^{a}=v-h^{a}(y)$. Here, $h^{a}(y)$ is a polynomial approximation of $h(y)$ of degree $r$, found by solving \eqref{partial} and they are related by $h(y)=h^{a}(y)+\mathcal{O}(\norm{y}^s)$ with $s>r$.  

Consider \eqref{comp_3} in the proof of Proposition \ref{second_new_proposition}. With the proposed event-triggering condition, $\norm{e_v}\leq \sigma(\norm{w^{a}}+\norm{y}^{(p+1)})$ is ensured throughout the implementation. Using the inequality $\norm{w^{a}}=\norm{w+(h(y)-h^{a}(y))}\leq\norm{w}+\norm{(h(y)-h^{a}(y))}\leq\norm{w}+\mathcal{O}(\norm{y}^s)$, we arrive at
\begin{equation}
\begin{aligned}
\dot{V}\leq&-k_6(1-s_y) \norm{y}^{p}-(1-s_f)\frac{\lambda_{min}(Q)}{2\sqrt{\lambda_{max}(P)}}\norm{w}\\
&+\mathcal{O}(\norm{y}^s).
\end{aligned}
\label{vdot_diff}
\end{equation}
The difference between \eqref{vdot_diff} and \eqref{vdot_or} in the proof of Proposition \ref{second_new_proposition} is the presence of $\mathcal{O}(\norm{y}^s)$. Close to the origin, the sum of the first and third term, which can be any polynomial in $\mathcal{O}(\norm{y}^s)$, is less than zero. Therefore $\dot{V}<0$ close to the origin and local asymptotic stability of the origin is guaranteed. The non-existence of Zeno behaviour can be shown along the same lines as in the proof of Proposition \ref{second_new_proposition}. Similar argument holds when using the triggering conditions presented in Propositions \ref{first_new_proposition} and \ref{third_new_proposition} with $w$ replaced with $w^a$.  

The insights gained above and in Section \ref{need} are formalised in the propositions to follow. The proofs of the Propositions are omitted, as they follow along the same lines as the proofs of Propositions \ref{first_new_proposition}-\ref{third_new_proposition} and differ as discussed in the present and the previous section.
\begin{proposition}
\label{first_corollary}
Consider the system \eqref{final_tranf}. Under the assumptions of Theorem \ref{center_manifold_theorem}, if the origin $y=0$ of the reduced system \eqref{dynamics_center} is locally asymptotically stable and there exists a LISS Lyapunov function such that in \eqref{ISS_char_1}, $\alpha_4 \in {\cal O}(\norm{y}^p)$, $p \leq 1$, then the origin of the overall system \eqref{final} is locally asymptotically stable under event-triggered implementations with relative triggering rules $\norm{e}\geq \sigma(\norm{y}+\norm{v})$ and $\norm{e}\geq \sigma(\norm{y}+\norm{w^{a}})$, with $\sigma$ chosen as in \eqref{sigma_sf}. Moreover, the inter-execution times $(t_{i+1}-t_i)$ are lower bounded by a positive constant for all $i \geq 0$.
\end{proposition}
The next proposition caters to the class of systems for which Assumption \ref{second_assumption} is satisfied.
\begin{proposition}\label{second_corollary}
Consider the system \eqref{final_tranf}. Under the assumptions of Theorem \ref{center_manifold_theorem}, if the origin $y=0$ of the reduced system \eqref{dynamics_center} is locally asymptotically stable and the conditions in Assumption \ref{second_assumption} are satisfied, then 
\begin{enumerate}
    \item If sets $S_3$ in \eqref{set_3} and $S_4$ in \eqref{final_ball} are such that $S_4 \subset S_3$, then the set $S_4$ containing the origin of the overall system \eqref{final} is locally asymptotically stable under the event-triggered implementation with relative thresholding $\norm{e_v}\geq \sigma(\norm{v}+\norm{y}^{(p+1)})$, with $\sigma$ chosen as in \eqref{sigma_sf}.
    \item The origin of the overall system \eqref{final} is locally asymptotically stable under the event-triggered implementation with relative thresholding $\norm{e_v}\geq \sigma(\norm{w^{a}}+\norm{y}^{p+1})$, with $\sigma$ chosen as in \eqref{sigma_sf}.
\end{enumerate}
Moreover, in both the cases above, the inter-execution times $(t_{i+1}-t_i)$ are lower bounded by a positive constant for all $i \geq 0$.
\end{proposition}
The event-triggered implementation in case 1) of Proposition \ref{second_corollary} must be as in \eqref{imp2} with no triggering in the interior of the set $S_4$ to ensure Zeno-free implementation. Unlike Proposition \ref{third_new_proposition}, the set for which asymptotic stability can be guaranteed cannot be made arbitrarily small. 
\begin{proposition}
\label{third_corollary}
Consider the system \eqref{final_tranf}. Under the assumptions of Theorem \ref{center_manifold_theorem}, if the origin $y=0$ of the reduced system \eqref{dynamics_center} is locally asymptotically stable and the conditions in Assumption \ref{third_assumption} are satisfied, then the trajectories of the system \eqref{final} are locally ultimately bounded by $\mathcal{B}_{r_s} \subset S_v$ (a sub-level set of \eqref{lyap_prop2} where \eqref{vdot_diff} holds), under the event-triggered implementation \eqref{imp2} with relative thresholding $\norm{e}\geq \sigma(\norm{w^{a}}+\norm{y}^{(p+1)})$ and $\sigma$ chosen  according to \eqref{sigma_sf}. Moreover, the inter-execution times $(t_{i+1}-t_i)$ are lower bounded by a positive constant for all $i \geq 0$.
\end{proposition}
Propositions \ref{first_corollary}-\ref{third_corollary} are analogues of the Propositions \ref{first_new_proposition}-\ref{third_new_proposition} from Section \ref{event_4}. Next, we implement the controller designed for Example 2 using the relative thresholding rule from Proposition \ref{second_corollary}. 
\subsection{Example 2 revisited}
The center manifold $h(y)$ was found up to order two and $h^{a}_1(y)=[y^2 \; -y^2]^{\top}$.  The Lyapunov function from Proposition \ref{second_new_proposition} is of the form $V=|y|+\sqrt{w^{\top}Pw}$, where $P\succ0$ is found by solving the Lyapunov equation for $Q=\mathbb{I}_2$. 
With the thresholding rule from Proposition \ref{second_corollary}, with $s_f=0.5$, we arrive at 
\begin{align*}
\dot{V}\leq -|y|^3-\frac{\lambda_{min}(Q)}{4\sqrt{\lambda_{max}(P)}}\norm{w}+{\cal O}(|y|)^4.
\end{align*}
Local asymptotic stability of the origin and the non-existence of Zeno behaviour is guaranteed through Proposition \ref{second_corollary}. 
\begin{figure}
\centering
\subfloat[Evolution of the states of the closed-loop system \label{plot1_good}]{\includegraphics[width=0.49\linewidth]{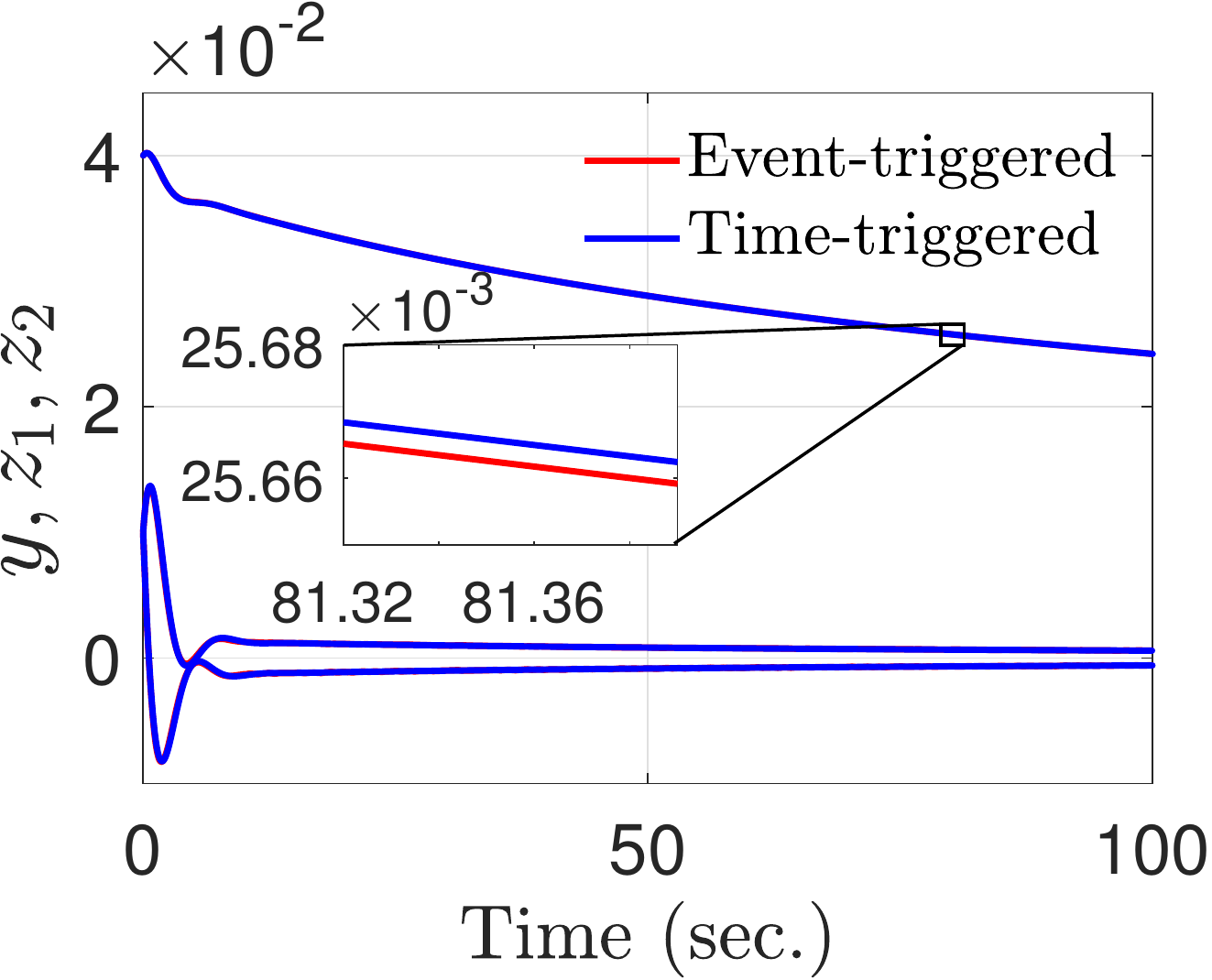}}
\hfill
\subfloat[Evolution of the system in the $y-z_1$ plane\label{plot2_good}]{\includegraphics[width=0.49\linewidth]{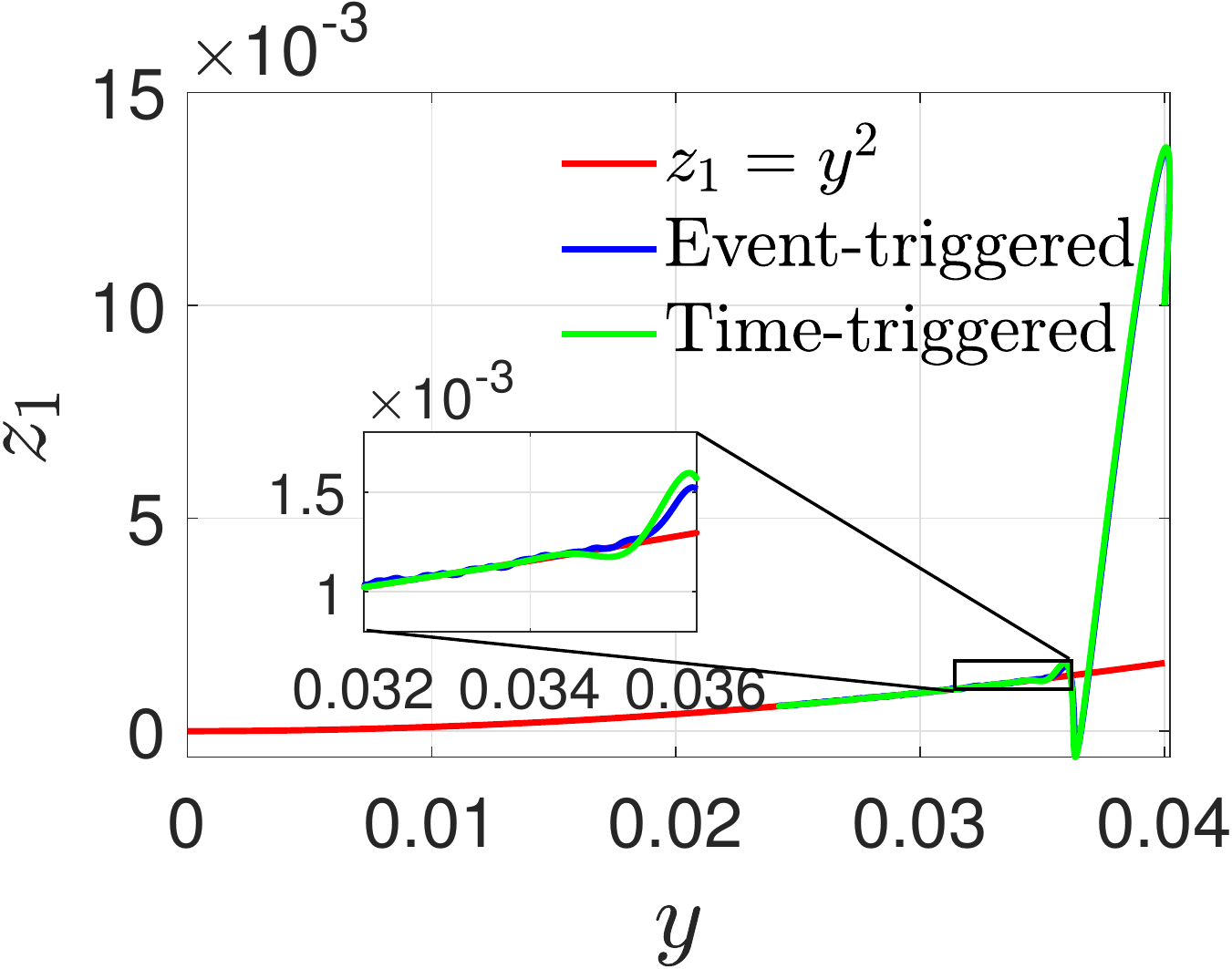}} 
\\ \vspace{0.2cm}
\subfloat[Evolution of the system in the $y-z_2$ plane\label{plot3_good}]{\includegraphics[width=0.49\linewidth]{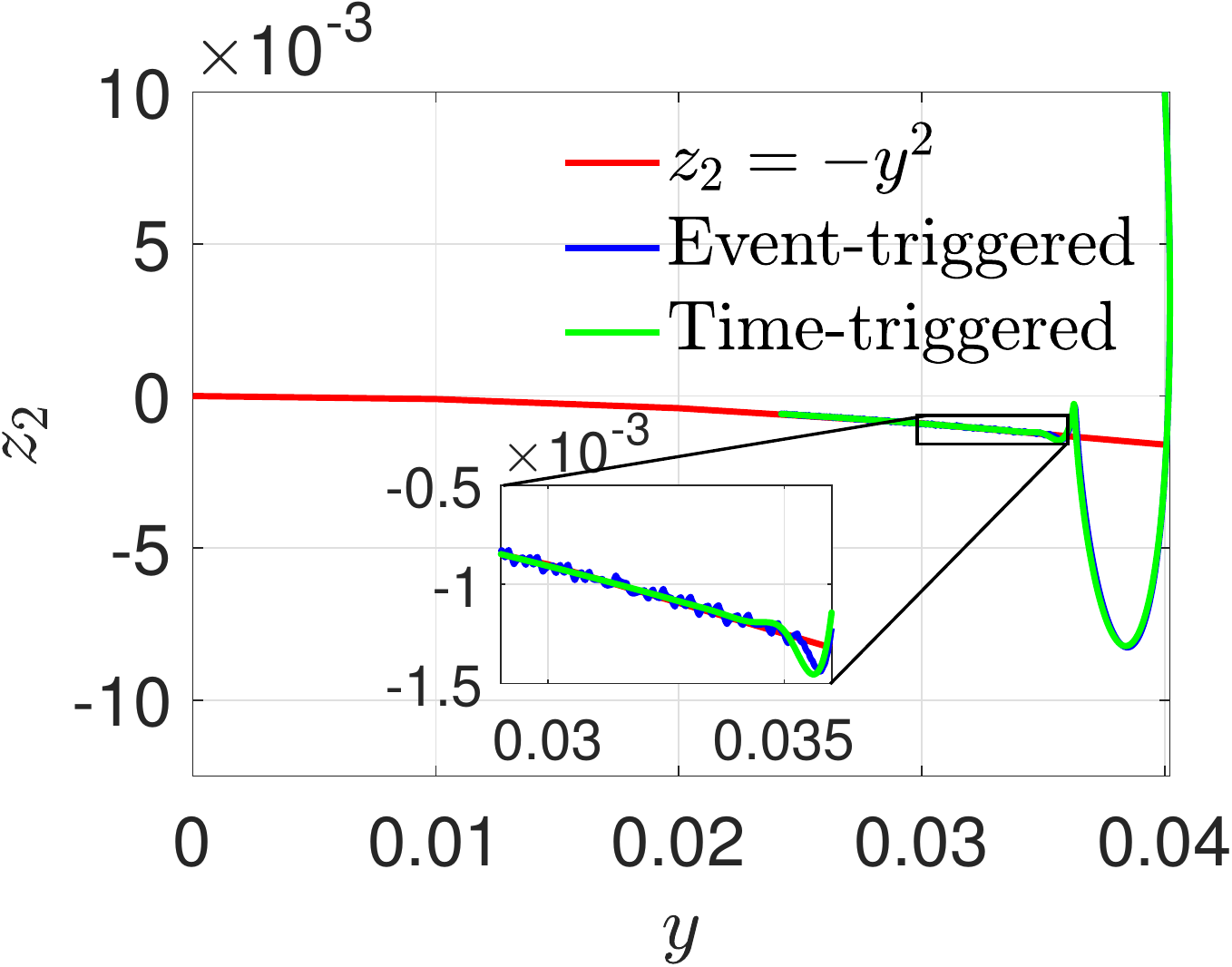}}
\hfill
\subfloat[Inter-event times\label{plot4_good}]{\includegraphics[width=0.49\linewidth]{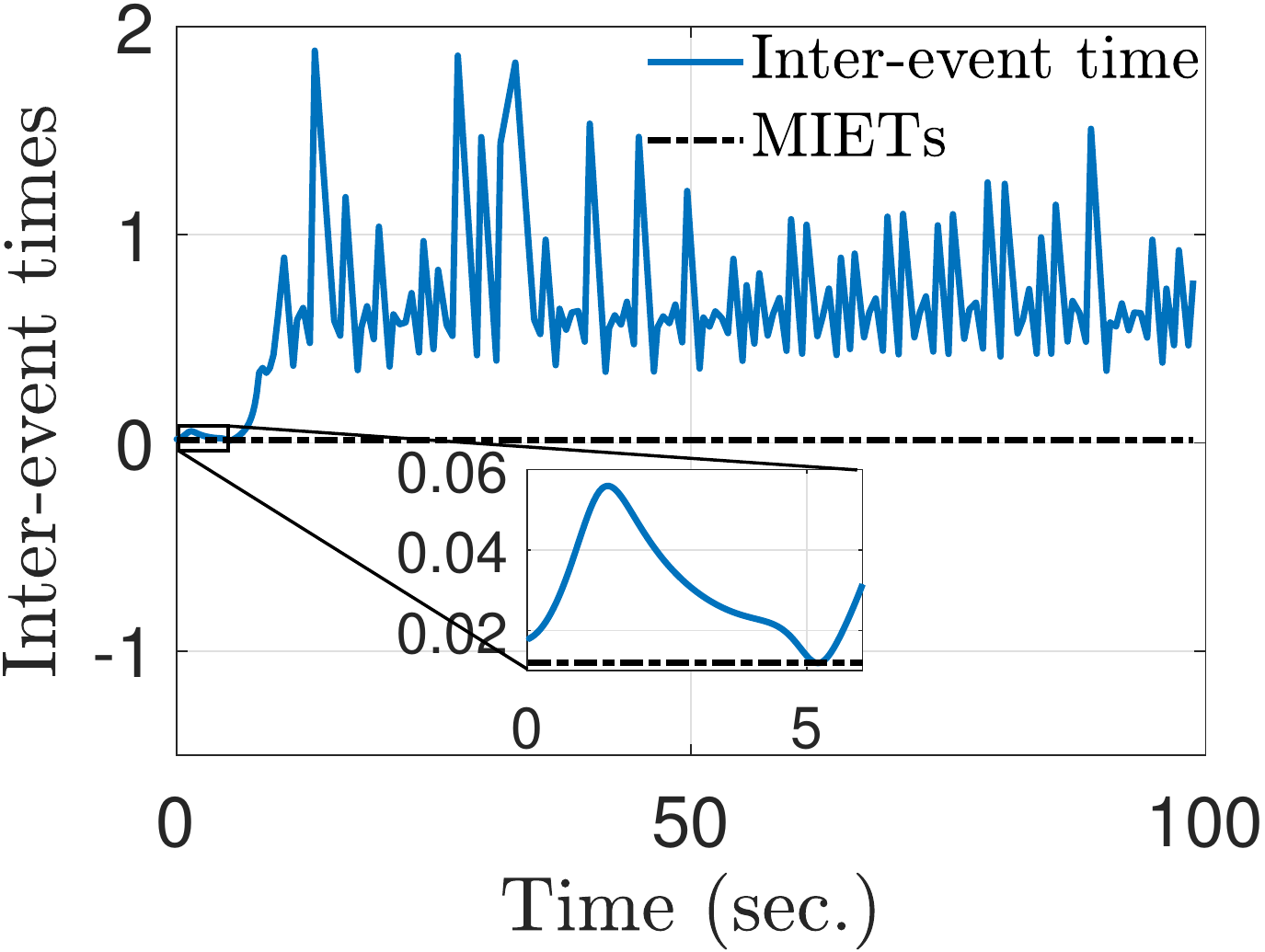}}
\caption{Simulation results for the event-triggered implementation of a controller designed for system \eqref{sys_ex2} with the thresholding rule $\norm{e_z}\geq 0.03(y^4+\norm{z-h^{a}_1(y)})$, $\sigma=0.03$.}
\label{Example_2_sim}
\end{figure}
Simulation results using the triggering condition $\norm{e_z}\geq 0.03(y^4+\norm{z-h^{a}_1(y)})$ with $\sigma=0.03$ chosen to satisfy \eqref{sigma_sf} are presented in Figure \ref{Example_2_sim}. As the transformation $v=z-Ey$ is not necessary for system \eqref{example_closed}, the variable $z$ appears in the triggering condition. The system is initialized at $(0.04,0.01,0.01)$. 
Exponential convergence of the states $z_1$ and $z_2$ to the center manifold can be seen in Figures \ref{plot2_good} and \ref{plot3_good}.  As in Example 1, minimum inter-event time (MIETs) from ten initial conditions was found to be \unit{11.9}{\milli\second}. The mean time between two events was found to be \unit{541.3}{\milli\second}. In Figure \ref{plot4_good}, we see that the inter-event times are lower bounded by MIETs. With MIETs as the sampling time, time-triggered implementation was performed and in Figure \ref{plot1_good} we see that the performance of the event-triggered system is a close match to that of the time-triggered system. On an average, the control is updated every \unit{541.3}{\milli\second} in event-triggered control and every \unit{11.9}{\milli\second} in time-triggered control, which makes a case in favour of event-triggered control. From Figures \ref{plot2_good} and \ref{plot3_good}, we see that the trajectories of the system show oscillations about the center manifold, as they tend to the origin along the center manifold. When the exact knowledge was available and was employed in checking the triggering condition in Example 1, the trajectories showed no oscillations along the center manifold.

Next, we analyse the effects of using better approximations of the center manifold in the triggering condition. Solving the PDE \eqref{partial} for $h(y)$ upto orders four and six we get 
\begin{align*}
h^a_2(y)=\left[\begin{matrix}
y^2\\ -y^2-10y^4
\end{matrix}\right] \;\;
h^a_3(y)=\left[\begin{matrix}
y^2-200y^6 \\ -y^2-10y^4-80y^6
\end{matrix}\right].
\end{align*}
Using the two approximations, simulations were performed for two sets of initial conditions, $P_1$ and $P_2$. The initial conditions in $P_1$ are close to the origin and those in $P_2$ are relatively farther. The initial conditions $(0.04,0.01,0.01)$ and $(0.4,0.1,0.1)$ are sample initial conditions from sets $P_1$ and $P_2$ respectively. As the approximation for $h(y)$ improves, it is found that, for initial conditions in both $P_1$ and $P_2$, there is no discernible difference in the system performance and the qualitative properties of trajectories. For initial conditions in $P_1$, no change in the triggering behaviour is observed but for initial conditions in $P_2$, the mean inter-event times using $h^{a}_1(y), h^a_2(y)$ and $h^a_3(y)$ are found to be $\unit{0.2021}{\second}$, $\unit{0.2018}{\second}$ and $\unit{0.2014}{\second}$ respectively. This indicates that use of a better approximation leads to more frequent triggers, as the triggering condition is checked in a more refined manner.

Next, we use the triggering condition from Proposition \ref{second_corollary} for the event-triggered implementation of a position stabilizing controller for the Mobile Inverted Pendulum (MIP) robot.
%
%
\subsection{Example 3 : Position stabilization of MIP robot} \label{C5_position}
The MIP robot is a four degree-of-freedom robot with two independently driven wheels and a pendulum-like central body that has unstable pitching motion under the influence of gravity. A controller that asymptotically drives the robot to the origin of the $(x,y)$ plane, while maintaining an upright position, is considered here for event-triggered implementation. There is no specification on the final orientation $\theta$ of the robot at the origin. This is called position and reduced attitude stabilization \cite{vijay_TCST}. For simplicity, we refer to this task as position stabilization.

\subsubsection{Mathematical model} 
The model of the robot from \cite{vijay_TCST} and \cite{maru1} is considered, with the states $x =(x_1,x_2,x_3,x_4,x_5,x_6)= 
(\tanh (x\sin\theta-y\cos\theta),\tanh (x\cos\theta+y\sin\theta),\alpha,
\dot{\alpha},v,\dot{\theta})^\top \in \mathcal{Q}_c \triangleq (-1,1)^2 \times  S^1 \times \mathbb{R}^3$, 
where $S^1$ denotes the unit circle, $(x,y)$ represents the position of the robot in the plane, $\theta$ is the orientation of the robot at this position, $\alpha$ and $\dot{\alpha}$ are the pitch and pitch velocity, $v$ and $\dot{\theta}=\omega$ are the linear and angular velocities of the robot. The state-space model of the MIP robot presented in \cite{vijay_TCST,maru1} is
%
%
\begin{align*}
&\hspace{2cm}\dot{x}=f_1(x)+h_1(x)u_1+h_2(x)u_2\\
&f_1(x)= \\ &\left(\begin{array}{c}
(1-x_1^2) (\tanh^{-1} x_2) x_6 \\
(1-x_2^2)(-(\tanh^{-1} x_1) x_6 + x_5) \\
x_4 \\
\dfrac{\left(\sin{(2x_3)}(b_3b_1 x_6^2-b_5^2 x_4^2)+2b_5b_3 a_g\sin{x_3} \right) }{m_1(x)} 
\\
\dfrac{\left(\sin{(2x_3)}(-b_1b_5 x_5^2\cos{x_3}-b_5^2 a_g)  +2b_2b_5 x_4^2\sin{x_3} \right)}{m_1(x)} 
\\
\frac{-b_1 x_4 x_6\sin{(2x_3)}}{m_2(x)} 
\end{array}\right)
\end{align*}
\begin{align*}
h_1(x) &=( 
0,0,0, 
\frac{-2(b_3 r+b_5\cos{x_3})}{m_1(x)} , 
\frac{2(b_2+b_5 r \cos{x_3})}{m_1(x)} , 0
 )\\  
h_2(x)&=\left( 
0 , 0 , 0 , 0 , 0 , \frac{b}{ m_2(x) r} 
 \right), \\
m_1(x) &= 2(b_2b_3-b_5^2\cos^2{x_3}) \text{\;and\;}
m_2(x) = (b_4+b_1\sin^2{x_3})
\end{align*}
where $m_1(x),m_2(x) > 0, \forall\ x\in \mathcal{Q}_c$, 
and $b, r, b_i,i=1,...,5$ are constants dependent on the robot parameters.
%
The control inputs are $u=(u_1,u_2)=(\tau_r+\tau_l,\tau_r-\tau_l)$, 
where $(\tau_r,\tau_l)$ are the torques applied to the right and left wheels respectively and 
$a_g$ is the acceleration due to gravity.
Linearizing about $(x;u)=(0;0)$ yields
\begin{align}
\dot{x} &=  Ax + 
\left[\begin{array}{c}
0\\B_1
\end{array} \right] u_1 + 
\left[\begin{array}{c}
0\\
B_2
\end{array} \right] u_2+\tilde{G}(x,u)     \label{lin_model1} 
\end{align}
where 
\begin{align*}
A = \left[ \begin{array}{ll}
0 & 0_{1\times 5} \\
0_{5\times 1} & A_1
\end{array} \right], \;\; 
A_1 = \left[ \begin{array}{lllll}
0 & 0 & 0 & 1 & 0\\
0 & 0 & 1 & 0 & 0\\
0 & a_1 & 0 & 0 & 0\\
0 & a_2 & 0 & 0 & 0\\
0 & 0 & 0 & 0 & 0
\end{array} \right],\\
B_1 = \begin{bmatrix}
0 & 0 & a_3 & a_4 & 0
\end{bmatrix}^{\top}, \; B_2 = \begin{bmatrix}
0 & 0 & 0 & 0 & a_5\end{bmatrix}^{\top}
\end{align*} 
and $\tilde{G}$ are the nonlinearities satisfying conditions in \eqref{conditions_1}.
%
The controllability matrix of the linearized system \eqref{lin_model1} has rank 
five, with $x_1$ being the uncontrollable state. 
\subsubsection{Choice of a controller}
A linear state-feedback control law with the following structure 
\beq
u_1=-{K}_1 [x_2\ x_3\ x_4\ x_5]^\top,\; u_2= -{K}_2 [x_6\ x_1]^\top \label{inner-loop}
\eeq
where, $K_1=[k_{1i}]_{i=1,\ldots,4}\in \mathbb{R}^4,
K_2=[k_{2j}]_{j=1,2}\in \mathbb{R}^2$ was proposed in \cite{vijay_TCST} to achieve the control objective of reduced attitude stabilization. For ease of notation, the controller is represented as $u=K(x)$. The controller parameters are chosen so as to assign negative real eigenvalues 
to the fifth-order controllable subsystem of \eqref{lin_model1}.  
The eigenvalue corresponding to the $x_1$ dynamics is zero and cannot be influenced by the controller \eqref{inner-loop}. The state $x_1$ is used in the control $u_2$ to  indirectly influence and stabilize the dynamics on the center manifold.  
\subsubsection{Center manifold analysis of the closed-loop system}
Using the partitioning of the states $p=({p}_1,{p}_2)\in \mathcal{Q}_c, {p}_1\triangleq x_1,
{p}_2\triangleq\left(x_2,x_3,x_4,x_5,x_6\right)$, we arrive at
\beq
\begin{aligned}
\dot{p}_1  &= {A}^c {p}_1 + G_1({p}_1,{p}_2)\\
\dot{p}_2  &= {A}^s {p}_2 + B_2k_{21}p_1+ G_2({p}_1,{p}_2,K(p_1,p_2)) 
\end{aligned}
\label{linclosed1}
\eeq
where, ${A}^s \in \mathbb{R}^{5\times 5}$ is Hurwitz, 
${A}^c = 0$, $G_1 \in \mathbb{R}^5$ and $ G_2 \in \mathbb{R}$. As $G_1$ is independent of $u$, the control directly influences only the $p_1$-subsystem.
 
As discussed in the problem formulation in Section \ref{formulation_4}, the cross-coupling linear term, $B_2k_{21}p_1$ between the $p_1$ and $p_2$ subsystems must be eliminated for the results from center manifold theory to hold. This term is removed using the change of variables $\bar{p}_2=p_2-Ep_1$, $E \in \mathbb{R}^{5\times 1}$, where the matrix $E$ is obtained by solving \eqref{finde} and the matrix $E$ is $E=[0\;0\;0\;0\;-\frac{k_{22}}{k_{21}}]^{\top}$. Using the notations
\begin{align*}
g_1(p_1,\bar{p}_2&+Ep_1))=G_1(p_1,\bar{p}_2+Ep_1)) \\
g_2(p_1,\bar{p}_2&+Ep_1,K(p_1;\bar{p}_2+Ep_1)) \\
=&G_2(p_1,\bar{p}_2+Ep_1,K(p_1;\bar{p}_2+Ep_1))\\
&-EG_1(p_1,\bar{p}_2+Ep_1,K(p_1;\bar{p}_2+Ep_1))
\end{align*}
we arrive at
\begin{equation}
\begin{aligned}
\dot{p}_1 &= A^cp_1 + g_1(p_1,\bar{p}_2+Ep_1)\\
\dot{\bar{p}}_2 &= A^s\bar{p}_2 + g_2(p_1,\bar{p}_2+Ep_1,K(p_1;\bar{p}_2+Ep_1))
\label{final_tranf_5}
\end{aligned}
\end{equation}
where $\bar{p}_2\triangleq\left(x_2,x_3,x_4,x_5,x_6+(k_{22}/k_{21})x_1\right)$. 
By Theorem \ref{existence}, there exists locally, a smooth map
$h:\mathbb{R} \rightarrow \mathbb{R}^5$ such that $\bar{p}_2 = h(p_1)$ is 
a center manifold for the system \eqref{final_tranf_5}.
The choice of 
$h({p}_1) = (c_1 p_1^2 + \mathcal{O}(|p_1|^4),\mathcal{O}(|p_1|^4),\mathcal{O}(|p_1|^4),-c_2 p_1^2+\mathcal{O}(|p_1|^4)
,c_3 p_1^3+\mathcal{O}(|p_1|^4))$ (derived in \cite{vijay_TCST}), 
where 
\begin{align}
c_1=\frac{k_{14}k_{22}}{k_{11}k_{21}}, \; c_2= \frac{k_{22}}{k_{21}}, \; c_3= -\frac{b_4 k_{14}k_{22}^3 r}{k_{11}k_{21}^4 b}
\label{cis}
\end{align}
qualifies as a local center manifold, since it satisfies \eqref{partial}
in a small neighbourhood of the origin $(p_1;\bar{p}_2)=0$.
The dynamics on the one-dimensional center manifold is governed by 
\begin{align}
\dot{p}_1 = (c_3 p_1^3 - c_2 p_1)(1-p_1^2)\tanh^{-1}(c_1^2 p_1^2)
+\mathcal{O}(|p_1|^{6}). \label{reduced_mip}
\end{align}
The local asymptotic stability of $p_1=0$ of \eqref{reduced_mip} can be inferred using the Lyapunov function $V_2 = \frac{1}{2}p_1^2$, which yields 
\begin{align}
\dot{V_2}=-c_2c_1^2p_1^4+k|p_1^6|<0 \;\; \forall \; |p_1|<c_1\sqrt{\dfrac{c_2}{k}}
\label{vdot_mip}
\end{align}
where the constant $k>0$. By Theorem \ref{red_theorem} and Proposition \ref{main_result}, local asymptotic stability and local input-to-state stability of the overall closed-loop system \eqref{lin_model1}-\eqref{inner-loop} respectively is concluded.
\subsubsection{Event-triggered implementation}
The controller \eqref{inner-loop} that places the closed-loop poles of the fifth order controllable subsystem of \eqref{lin_model1} at $\{-2 \pm j0.5, -0.9, -0.5, -1\}$ is
\begin{equation}
\begin{aligned}
u_1&=0.1091x_2+7.0089x_3+1.0014x_4+0.4302x_5 \\ u_2&=-0.1929x_6-0.09645x_1.
\end{aligned} 
\label{part_control}
\end{equation}
Th constants $c_1,c_2$ and $c_3$ in \eqref{cis} are functions of the robot parameters and the controller gains $k_{ij}$. For the controller chosen, it can be checked that the dynamics on the reduced system \eqref{reduced_mip} is locally asymptotically stable. The conditions in Assumption \ref{second_assumption} are satisfied and we use the relative thresholding $\norm{e_{\bar{p}_2}}\geq \sigma (\norm{\bar{p}_2-h^{a}(p_1)}+\norm{p_1}^4)$ from Proposition \ref{second_new_proposition} with $h^{a}(p_1)\triangleq[c_1 p_1^2,\; 0,\; 0,\; -c_2 p_1^2,\; c_3 p_1^3]^{\top}$ for event-triggered implementation. 
Under this implementation, asymptotic stability and the non-existence of Zeno behaviour is guaranteed through Proposition \ref{second_corollary}.

The pair of matrices $(P,Q)$ required to compute the bound on the thresholding parameter $\sigma$ from \eqref{sigma_sf} are found by solving the Lyapunov equation ${A^{s}}^{\top}P+PA^{s}+Q=0$ by choosing $Q=\mathbb{I}_5$.
For the controller \eqref{part_control} and the $(P,Q)$ pair chosen as above, the bound for the relative threshold is found to be $\sigma\leq 10^{-4}$. Choosing $\sigma=10^{-4}$, the simulation results of event-triggered position stabilization of the MIP robot are presented in Figures \ref{Chap5_sim4} and \ref{Chap5_sim6}. 

The MIP robot is initialized at $(x,y,\theta)=(2,2,\frac{\pi}{2})$ with the pitch upright, that is, $\alpha=\dot{\alpha}=0$ and $v=\dot{\theta}=0$. The evolution of the position of the robot is shown in Figure \ref{sim4a}. The robot asymptotically reaches the origin of the $(x, y)$ plane. In Figure \ref{sim4b}, the evolution of the norm of the error $\norm{e_{\bar{p}_2}}$ and the threshold $10^{-4} (\norm{\bar{p}_2-h^{a}(p_1)}+\norm{p_1}^4)$ are shown. An event occurs when the norm of the error rises from zero to meet the threshold. In Figure \ref{sim4c}, we see that the pitch angle $\alpha$ and pitch velocity $\dot{\alpha}$ tend asymptotically to zero as guaranteed by Proposition \ref{second_corollary}. The evolution of the inter-event times is shown in Figure \ref{sim4d}. The inter-event times are lower bounded by $\unit{2.4}{\milli\second}$. Considering ten initial conditions inside a circle of radius \unit{3}{\meter} in the $(x, y)$ plane, it is found that the minimum time between two events (MIETs) is $\unit{2.4}{\milli\second}$. With MIETs as the sampling time for time-triggered control, simulations were performed and as can be seen in Figure \ref{sim4a}, the performance of event-triggered control is a close match to the performance of time-triggered control. In achieving this close similarity in performance, event-triggered control requires fewer closings of the control-loop than time-triggered control. From Figures \ref{plotcenter1_good}, \ref{plotcenter2_good} and \ref{plotcenter3_good}, we see that the trajectories of the event-triggered closed-loop system converge rapidly to the center manifold, while evolving slowly along the center manifold.
%
%
\begin{figure}
    \centering
  \subfloat[Position $(x,y)$ of the robot.\label{sim4a}]{%
       \includegraphics[width=0.49\linewidth]{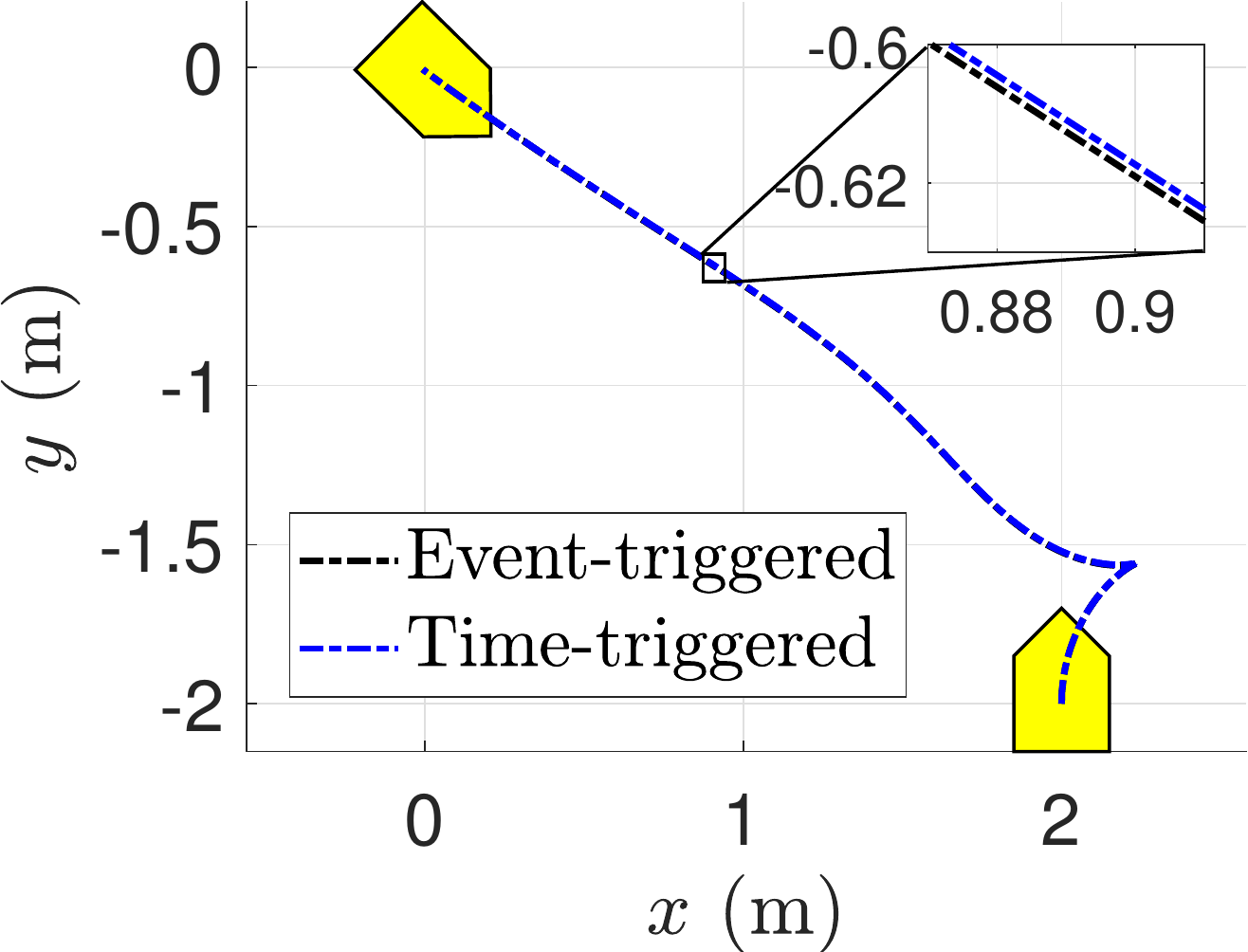}}
    \hfill
  \subfloat[Evolution of the error and the threshold\label{sim4b}]{%
        \includegraphics[width=0.49\linewidth]{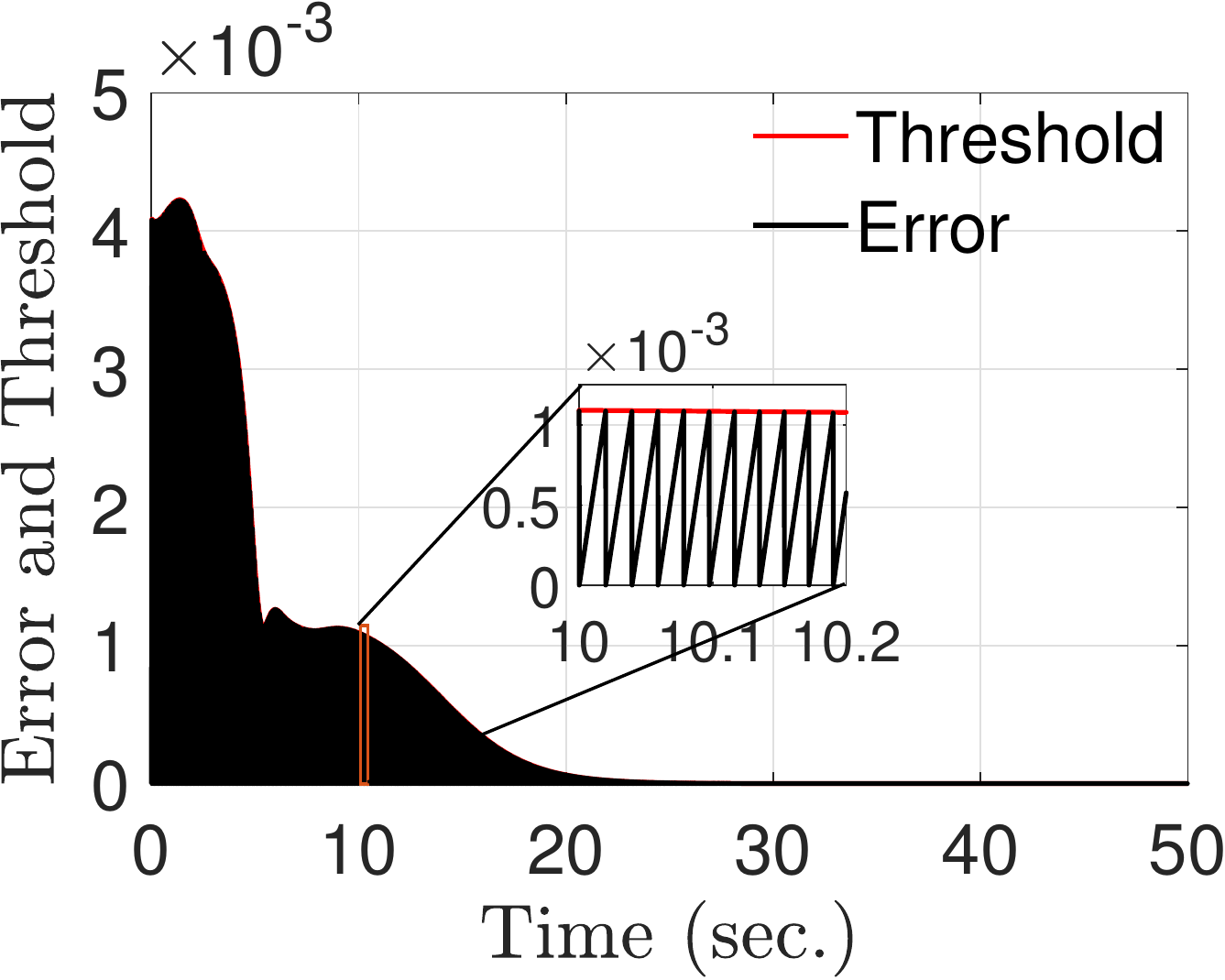}}
    \\ \vspace{0.5cm}
  \subfloat[Evolution of $\alpha$ and $\dot{\alpha}$ \label{sim4c}]{%
        \includegraphics[width=0.49\linewidth]{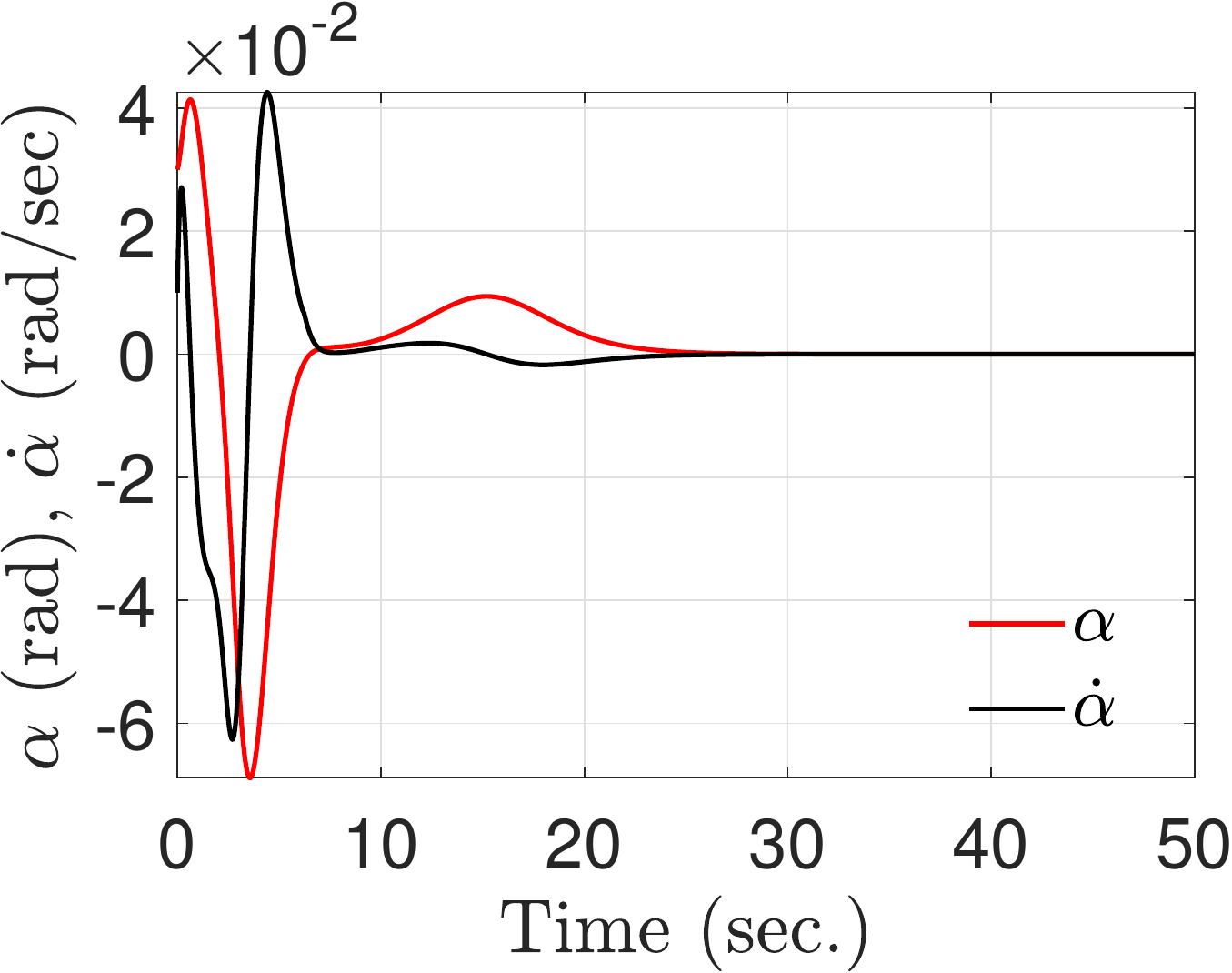}}
    \hfill
  \subfloat[Inter-event times \label{sim4d}]{%
        \includegraphics[width=0.49\linewidth]{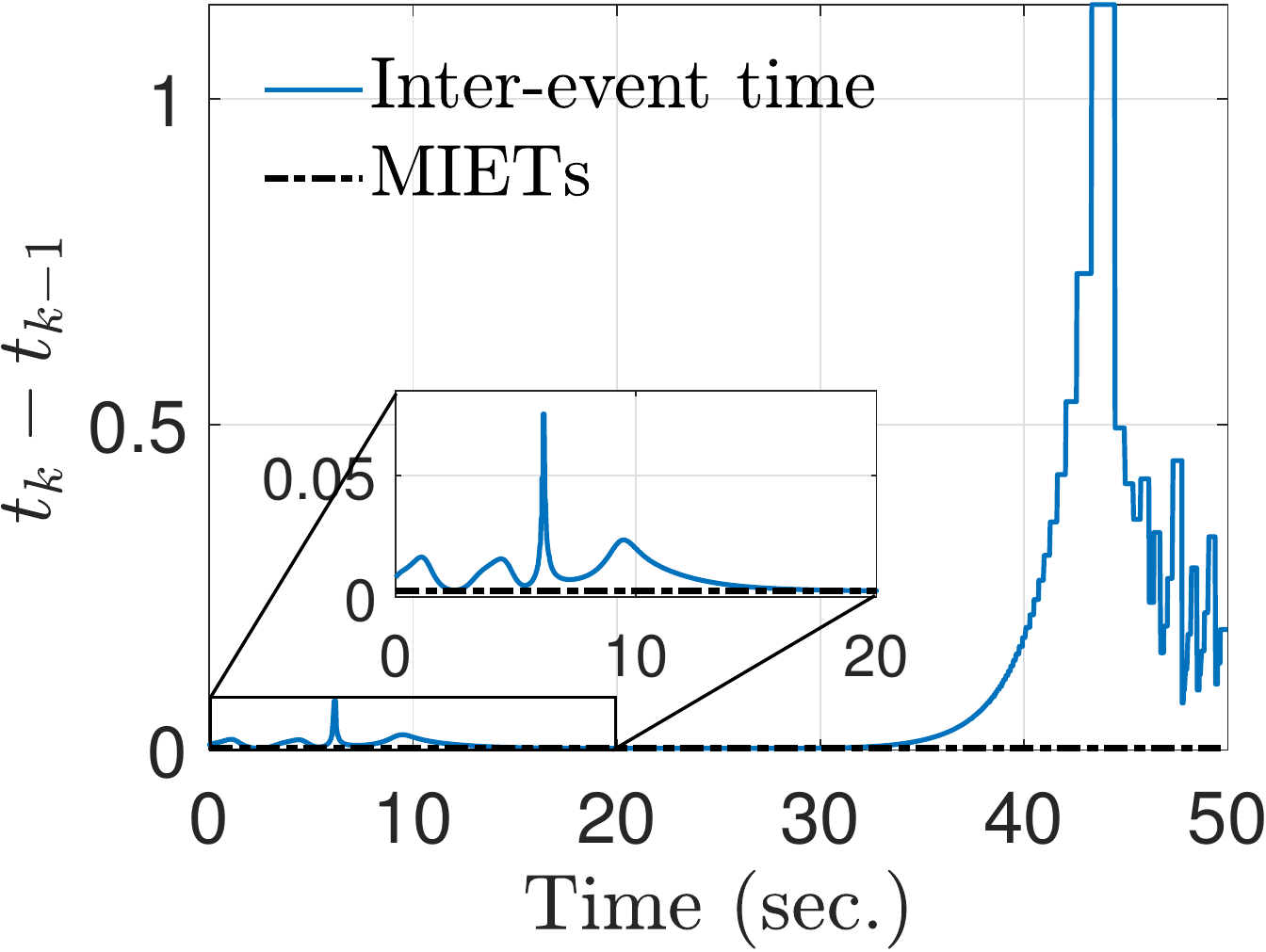}}
  \caption{Simulation results for the event-triggered position stabilization of an MIP robot.}
  \label{Chap5_sim4}
\end{figure}
\begin{figure} 
    \centering
  \subfloat[Evolution in $(x_1,x_2)$ plane \label{plotcenter1_good}]{%
       \includegraphics[width=0.49\linewidth]{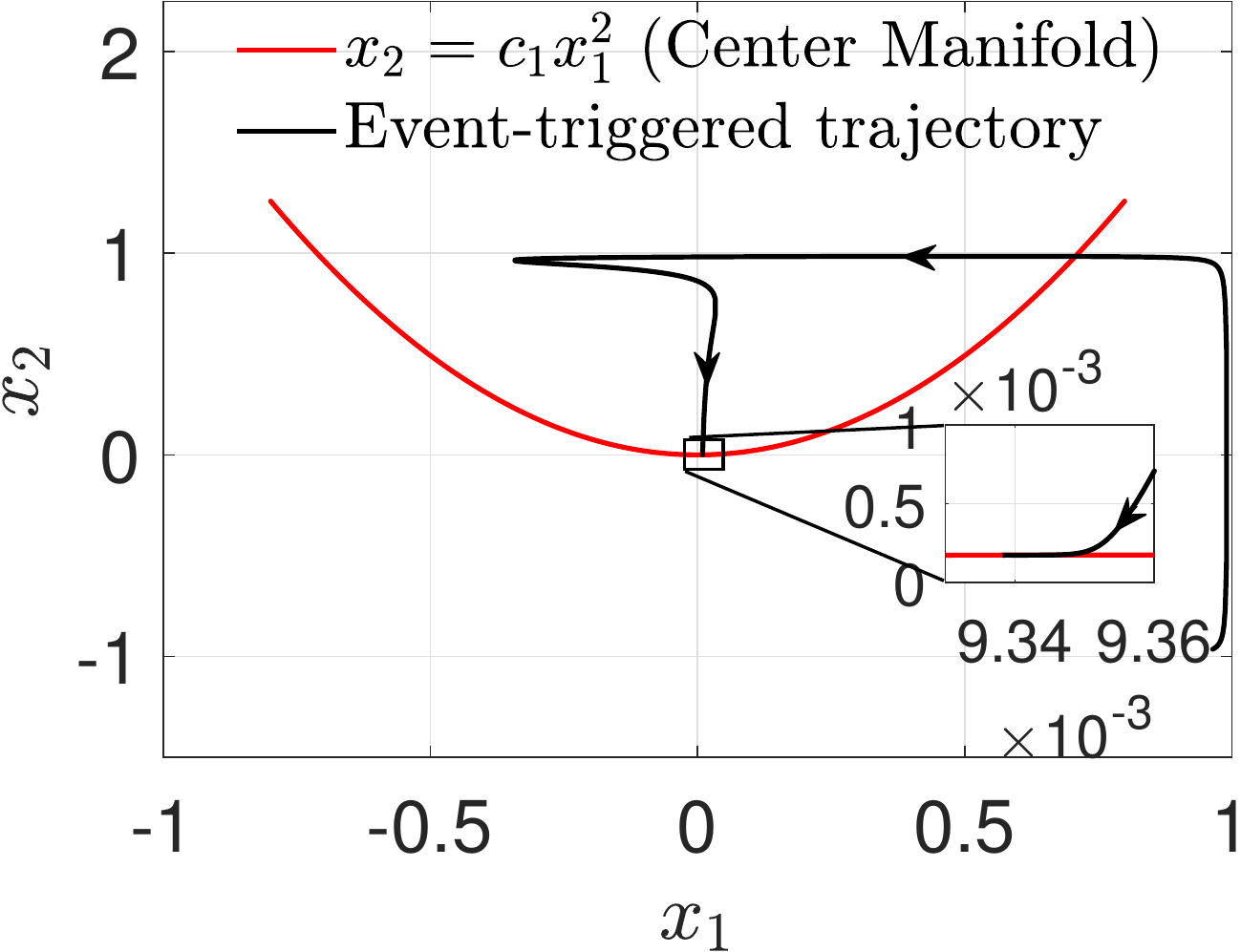}}
    \hfill
  \subfloat[Evolution in $(x_1,x_5)$ plane \label{plotcenter2_good}]{%
       \includegraphics[width=0.49\linewidth]{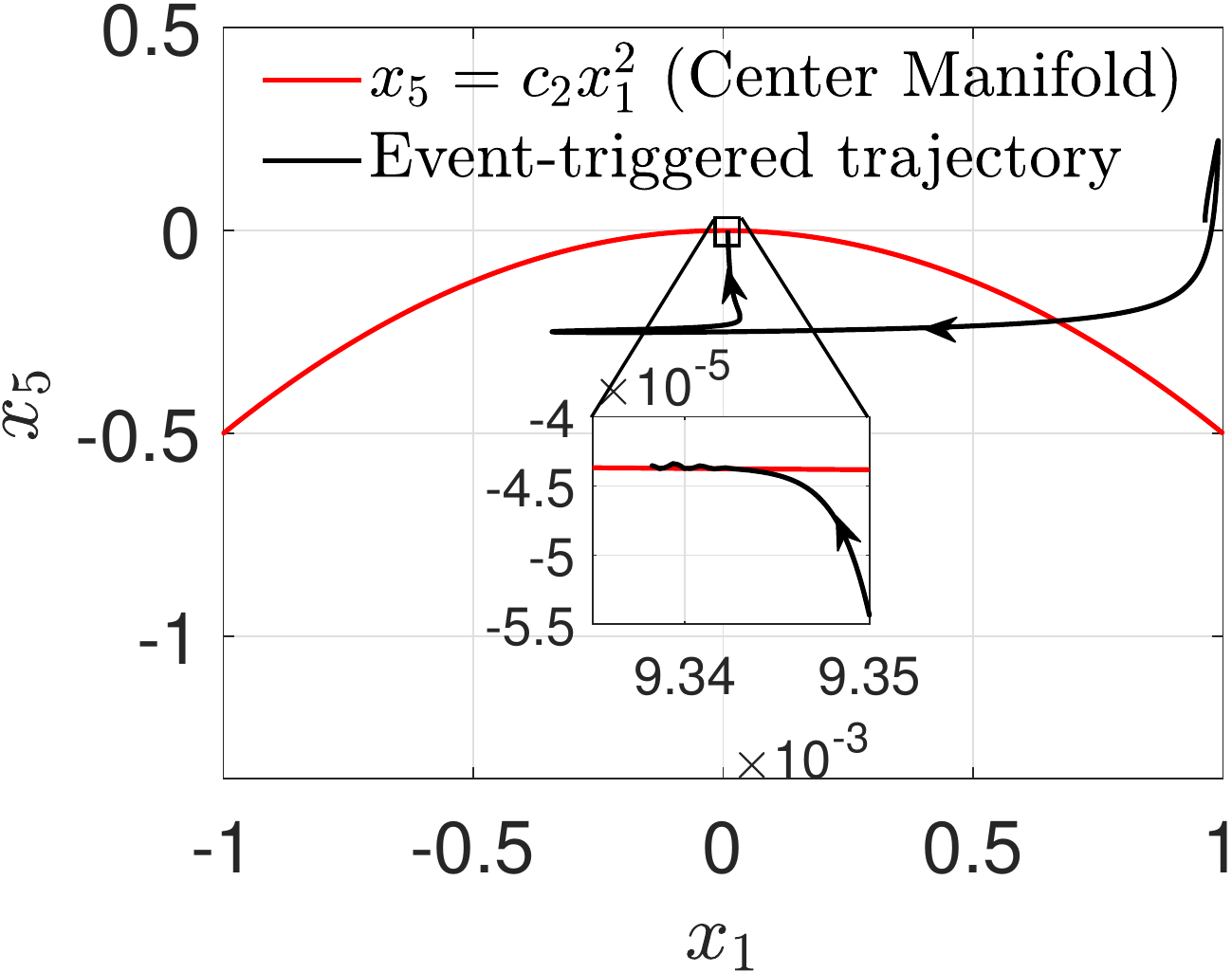}}
    \\ \vspace{0.5cm}
  \subfloat[Evolution in $(x_1,\bar{p}_2(5))$ plane \label{plotcenter3_good}]{%
        \includegraphics[width=0.47\linewidth]{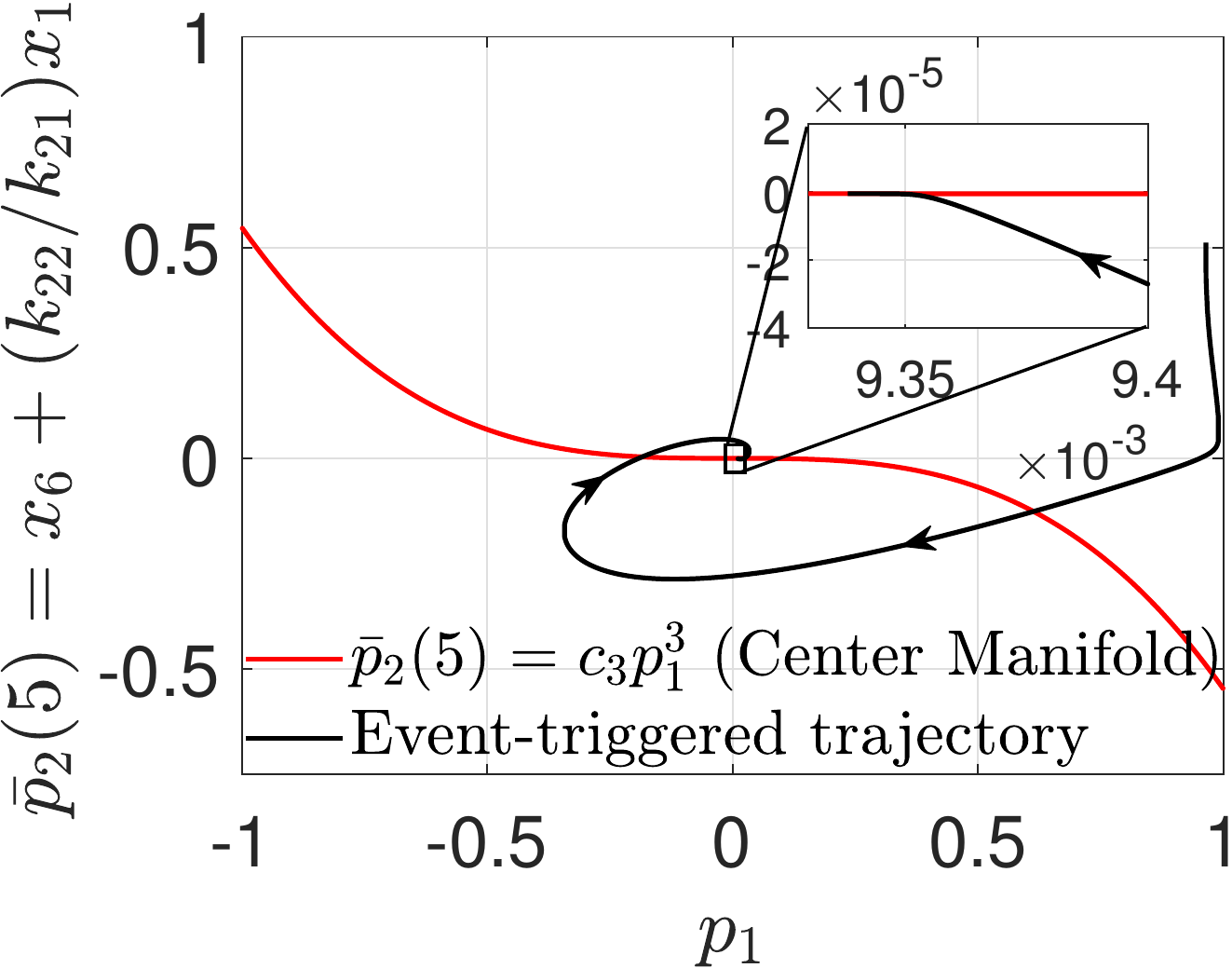}}
  \caption{The trajectories of the states $x_2, x_5$ and $\bar{p}_2(5)=x_6+(k_{22}/k_{21})x_1$ of the event-triggered closed-loop system.}
  \label{Chap5_sim6}
\end{figure}
\section{Conclusions}
\label{conclusions}
In this work, we presented event-triggered implementation of control laws designed for nonlinear systems with center manifolds. The proposed methods ensured Zeno-free local ultimate boudedness of the trajectories, and under some assumptions on the controller structure, Zeno-free asymptotic stability of the origin. Systems for which the center manifold is exactly computable were considered first and triggering conditions were presented, the checking of which requires the exact knowledge of the center manifold. Then, we considered systems for which the center-manifold can only be approximately computed and showed that the same triggering conditions could be used with the available approximate knowledge of the center manifold. 

When the triggering conditions employed the approximate knowledge of the center manifold, the trajectories of the closed-loop system showed oscillations about the center manifold, as they tend to the origin along the center manifold. When the exact knowledge was employed, the trajectories converged to the origin without any oscillations about the center manifold. For initial conditions close to the origin, a simulation study showed that use of a better approximation of the center manifold leads to more frequent triggers (as the triggering condition is checked in a more refined manner) and yields no significant improvement in the system performance. 
The minimum inter-event time from multiple simulations (MIETs) was used as sampling time for time-triggered control and it was found that event-triggered control yields similar performance as time-triggered control but with significantly lesser control updates. 

As part of our future work, we look to relax the assumptions on the controller structure, to guarantee Zeno-free asymptotic stability of the origin for a larger class of nonlinear systems with center-manifolds.
\bibliographystyle{ieeetr}
\bibliography{refs_Thesis}
\vskip 0pt plus -1fil
\vspace{-0.75cm}
\begin{IEEEbiography}[{\includegraphics[width=1in,height=1.25in,clip,keepaspectratio]{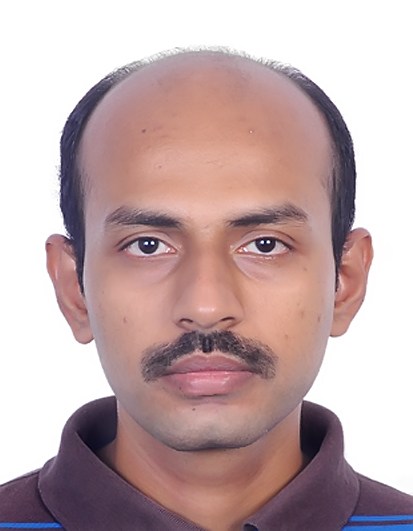}}]{Akshit Saradagi}
is currently a Postdoctoral researcher in the Department of Electrical Engineering, Indian Institute of Technology Madras, Chennai, India. He received the B.E degree in Electronics and Communication Engineering from P.E.S Institute of Technology, Bengaluru in 2013 and the M.S and Ph.D Dual degree from Indian Institute of Technology Madras, Chennai in 2020. His current research interests are in the areas of networked dynamical systems, control under resource constraints and bio-inspired control.
\end{IEEEbiography}
\vskip 0pt plus -1fil
\vspace{-0.85cm}
\begin{IEEEbiography}[{\includegraphics[width=1in,height=1.25in,clip,keepaspectratio]{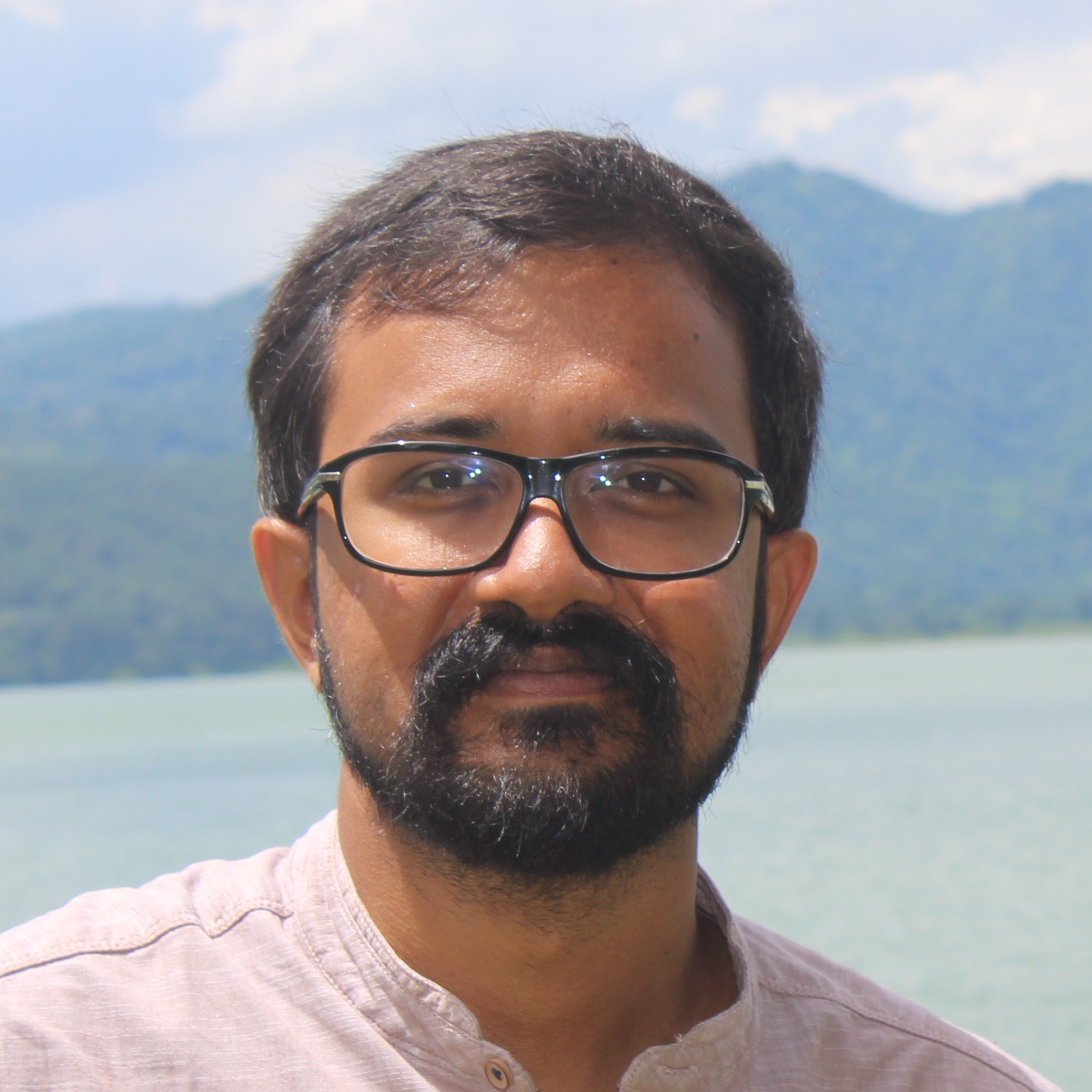}}]{Vijay Muralidharan}
obtained his PhD in robotics and control from IIT Madras, India. His experience includes postdoc and product development with LTU Sweden and SigTuple Technologies India, respectively. He is currently an Assistant Professor in the Department of Electrical Engineering at Indian Institute of Technology Palakkad, India. His research interests include dynamics modeling, nonlinear control, state estimation for robotic systems, medical devices and aerospace systems. 
\end{IEEEbiography}
\vskip 0pt plus -1fil
\vspace{-0.85cm}
\begin{IEEEbiography}[{\includegraphics[width=1in,height=1.25in,clip,keepaspectratio]{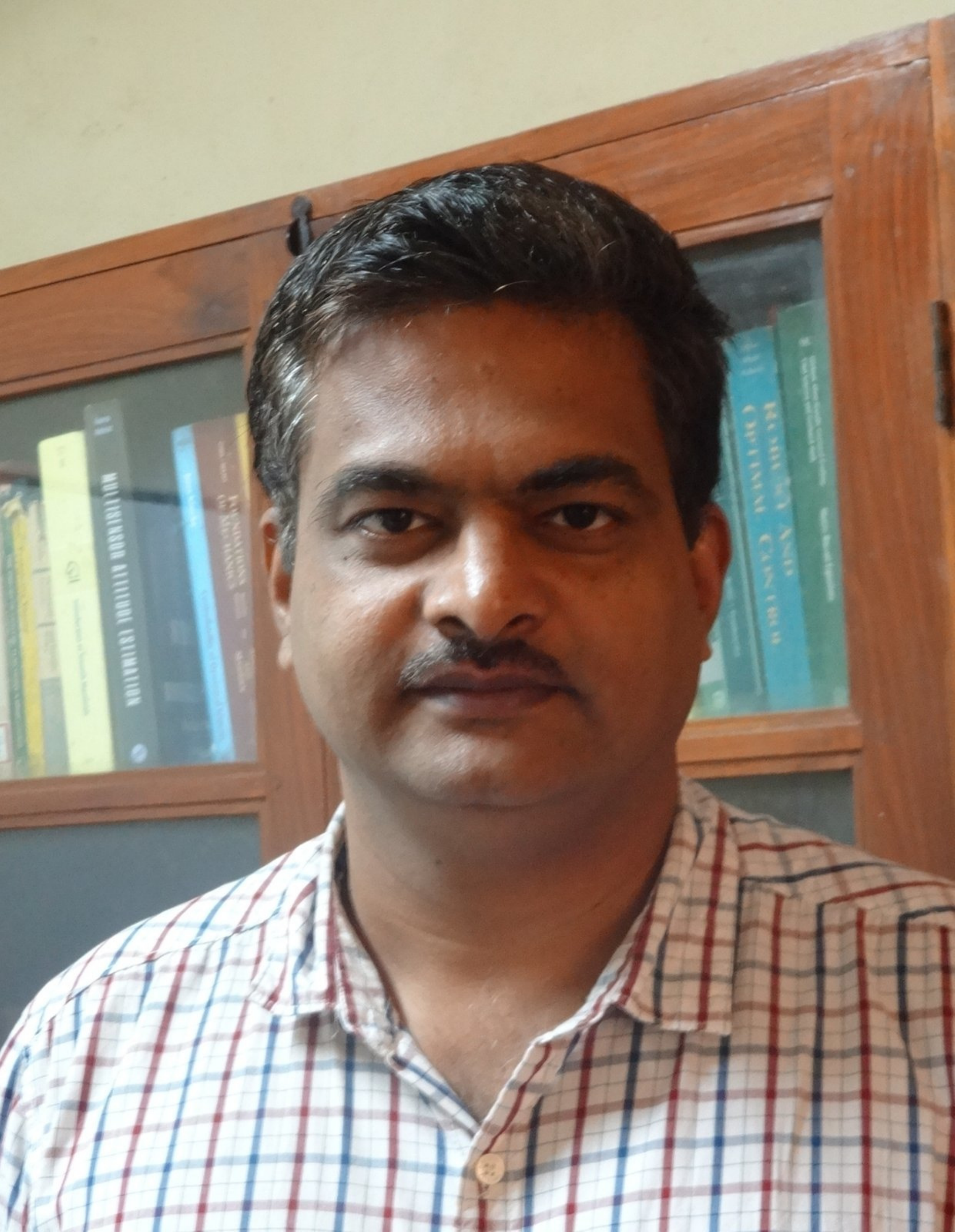}}]{Arun D Mahindrakar}(M'93-SM'19)
received the Bachelor of Engineering in Electrical and Electronics from Karnatak University, Dharwad, India in 1994, the Master of Engineering in Control Systems from VJTI, Mumbai in 1997 and the Ph.D in Systems  \& Control  from IIT Bombay in 2004. He worked as a Postdoctoral Fellow at the Laboratory of Signals and Systems at Sup\'elec, Paris from 2004 to 2005. He is currently a Professor in the Department of Electrical Engineering, Indian Institute of  Technology Madras, India. His research interests include  nonlinear stability, geometric control and  nonsmooth optimization.
\end{IEEEbiography}
\vskip 0pt plus -1fil
\vspace{-0.75cm}
\begin{IEEEbiography}[{\includegraphics[width=1in,height=1.25in,clip,keepaspectratio]{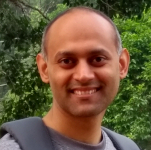}}]{Pavankumar Tallapragada}(S'12-M'14)
received the B.E. degree in   Instrumentation Engineering from SGGS Institute of Engineering $\&$ Technology, Nanded, India in 2005, M.Sc. (Engg.) degree in Instrumentation from the Indian Institute of Science in 2007 and the Ph.D. degree in Mechanical Engineering from the University of Maryland, College Park in 2013. He was a Postdoctoral Scholar in the Department of Mechanical and Aerospace Engineering at the University of California, San Diego from 2014 to 2017. He is currently an Assistant Professor in the Department of Electrical Engineering and the Robert Bosch Centre for Cyber Physical Systems at the Indian Institute of Science. His research interests include networked control systems, distributed control, multi-agent systems and  networked transportation systems.
\end{IEEEbiography}
\end{document}